\documentclass[11pt]{article}      %
\usepackage{authblk}
\usepackage{graphicx}%
\usepackage{multirow}%
\usepackage{amsmath,amssymb,amsfonts}%
\usepackage{mathrsfs}%
\usepackage[title]{appendix}%
\usepackage{xcolor}%
\usepackage{textcomp}%
\usepackage{manyfoot}%
\usepackage{booktabs}%
\usepackage{algorithm}%
\usepackage{algorithmic}%
\usepackage{listings}%
\usepackage{subcaption}

\usepackage{stmaryrd}

\usepackage{amsthm}
\numberwithin{equation}{section}
\newtheorem{definition}{Definition}
\newtheorem{lemma}{Lemma}
\newtheorem{theorem}{Theorem}
\newtheorem{remark}{Remark}
\newtheorem{proposition}{Proposition}

\begin{document}

\title{Accelerating Randomized Algorithms for Low-Rank Matrix Approximation
}

\author[1]{Dandan Jiang\thanks{\texttt{jiangdd@mail.xjtu.edu.cn}}}
\author[1]{Bo Fu\thanks{Corresponding author: \texttt{311342@stu.xjtu.edu.cn}}}
\author[2]{Weiwei Xu\thanks{Corresponding author: \texttt{wwxu@nuist.edu.cn}}}

\affil[1]{School of Mathematics and Statistics, Xi'an Jiaotong University, Shaanxi, 710049, China}
\affil[2]{School of Mathematics and Statistics, Nanjing University of Information Science and Technology, Nanjing 210044, China. The Peng
Cheng Laboratory, Shenzhen 518055, China, and with the Pazhou Laboratory (Huangpu), Guangzhou 510555, China.}
\maketitle

\begin{abstract}
Randomized algorithms are overwhelming methods for low-rank approximation that can alleviate the computational expenditure with great reliability compared to deterministic algorithms. A crucial thought is generating a standard Gaussian matrix $\mathbf{G}$ and subsequently obtaining the orthonormal basis of the range of $\mathbf{AG}$ for a given matrix $\mathbf{A}$. Recently, the \texttt{farPCA} algorithm offers a framework for randomized algorithms, but the dense Gaussian matrix remains computationally expensive. Motivated by this, we introduce the standardized Bernoulli, sparse sign, and sparse Gaussian matrices to replace the standard Gaussian matrix in \texttt{farPCA} for accelerating computation. These three matrices possess a low computational expenditure in matrix-matrix multiplication and converge in distribution to a standard Gaussian matrix when multiplied by an orthogonal matrix under a mild condition. Therefore, the three corresponding proposed algorithms can serve as a superior alternative to fast adaptive randomized PCA (\texttt{farPCA}). Finally, we leverage random matrix theory (RMT) to derive a tighter error bound for \texttt{farPCA} without shifted techniques. Additionally, we extend this improved error bound to the error analysis of our three fast algorithms, ensuring that the proposed methods deliver more accurate approximations for large-scale matrices. Numerical experiments validate that the three algorithms achieve asymptotically the same performance as \texttt{farPCA} but with lower costs, offering a more efficient approach to low-rank matrix approximation.
\end{abstract}

\section{Introduction}\label{sec1}

Low-rank matrix approximation, which aims to represent a matrix using another low-rank matrix, is fundamental in matrix computations and wields a key influence in machine learning, data analysis, and computer science. Deterministic techniques involve rank-revealing QR (RRQR) decomposition \cite{chan1987rank,gu1996efficient} and partial singular value decomposition (SVD) \cite{golub2013matrix}. Due to the interaction of rank and error, for a given matrix $\mathbf{A}$, two problems are typically considered in low-rank approximation: the fixed-precision approximation problem, seeking a matrix $\mathbf{B}$ with minimal rank such that $\interleave\mathbf{A}-\mathbf{B}\interleave\leq \varepsilon$ for a particular precision tolerance $\varepsilon$  and an appropriate norm $\interleave\cdot\interleave$; the fixed-rank approximation problem, aiming to seek a matrix $\mathbf{B}$ whose error $\interleave\mathbf{A}-\mathbf{B}\interleave$ is as small as possible and whose rank $k$ is fixed.

Restricted by computing power and storage capacity, as well as challenged by their difficulty in supporting parallel processing, deterministic techniques are not applicable within modern advanced computing architectures. Therefore, the randomized algorithms developed for low-rank approximation later have higher efficiency and are appropriate for use in parallel computing \cite{feng2023fast,gu2015subspace,halko2011finding,hallman2022block,kaloorazi2018compressed,kaloorazi2018subspace,liberty2007randomized,mahoney2009cur,martinsson2016randomized,rokhlin2010randomized,yu2018efficient}. Among these randomized algorithms, the randomized SVD algorithm \cite{halko2011finding} and its modified algorithms \cite{feng2023fast,gu2015subspace,hallman2022block,kaloorazi2018compressed,kaloorazi2018subspace,rokhlin2010randomized,martinsson2016randomized,yu2018efficient} are frequently employed. The fundamental idea behind these techniques is to identify an orthonormal basis for the given matrix, which involves a random test matrix for embedding, typically chosen as a standard Gaussian matrix. Within these Gaussian-based algorithms, a randomized framework of QB approximation for fixed-precision problem, termed \texttt{randQB\_EI}, was proposed in \cite{yu2018efficient}, offering an efficient error indicator to enable the adaptive termination once the approximation error achieves the desired threshold. In \cite{hallman2022block}, an alternative algorithm termed as \texttt{randUBV} was introduced, employing block Lanczos bidiagonalization to offer a solution with superior temporal efficiency relative to \texttt{randQB\_EI}. However, \texttt{randUBV} lacks the convergence bounds when the block size is smaller than the number of repeated singular values, and it also falls short in adaptability for achieving greater accuracy due to the absence of an instrument akin to the power iteration in \texttt{randQB\_EI}. 
To improve \texttt{randQB\_EI} and further combine it with the shifted power iteration technique, the \texttt{farPCA} algorithm was presented in \cite{feng2023fast}. Compared to \texttt{randUBV}, \texttt{farPCA} demands a similar computational cost while retaining the flexibility to achieve superior approximations with additional power iterations.

The standard Gaussian matrix is widely used for its special benefits, such as rotational invariance, concentration inequalities, norm bounds of its pseudo-inverse, etc. (see \cite{gu2015subspace}), which provides a theoretical foundation for the error analysis of Gaussian-based methods. Recently, a fast randomized algorithm using Bernoulli matrices as random test matrices has been proposed, which is less costly due to the sparse properties of Bernoulli matrices \cite{xu2023fast}. 
Similar random test matrices, such as subsampled random Fourier transform \cite{woolfe2008fast} and the sparse sign matrix \cite{martinsson2020randomized,xie2023sketchne}, are also available to speed up matrix multiplications. 
However, these matrices lack various advantageous properties of Gaussian matrices, making error analysis more challenging and limiting performance when overly sparse. 
To enhance performance and ensure robust theoretical analysis, we propose some modified test matrices to accelerate computations.

In this paper, we concentrate exclusively on the fixed-precision approximation problem and compare the performances of algorithms resulting from different random test matrices. Notably, when the input target matrix is sparse, leveraging these sparse random test matrices for acceleration yields minimal benefit. Therefore, this paper primarily focuses on accelerating approximation for dense matrices.

Our main contribution is the development of acceleration algorithms along with their corresponding theoretical analysis. First, the asymptotic behavior of the standardized Bernoulli matrix is analyzed when multiplied by an orthogonal matrix assuming a mild condition (condition \eqref{eq3.1}), where this condition can be interpreted as compensating for the trade-off between computational efficiency and stability. Second, analogous properties are exhibited in the sparse sign matrix, the sparse Gaussian matrix, and further in a series of random matrices whose elements are $i.i.d.$ with mean $0$ and variance $1$. Consequently, a set of corresponding acceleration algorithms is introduced, each utilizing different random test matrices: \texttt{farPCA} based on the standardized Bernoulli matrix (\texttt{farPCASB}), the sparse sign matrix (\texttt{farPCASS}), and the sparse Gaussian matrix (\texttt{farPCASG}). Under condition \eqref{eq3.1}, these algorithms achieve asymptotically the same performance as \texttt{farPCA} but with lower computational costs in theory. Numerical experiments corroborate the theoretical insights. Third, for the error analysis of the three algorithms, we first derive asymptotic norm bounds for the pseudo-inverse of the standard Gaussian matrix predicated on the random matrix theory (RMT) and elucidate an asymptotic error analysis for \texttt{farPCA} by employing these norm bounds, then the error analysis in \cite{halko2011finding} and asymptotic error analysis for \texttt{farPCA} can be reamed to the three acceleration algorithms. The asymptotic error analysis yields more precise error bounds, which enables us to manage the approximation error when dealing with large matrices effectively. Notably, similar technique can also be employed to improve other Gaussian-based random algorithms in \cite{gu2015subspace,halko2011finding,hallman2022block,kaloorazi2018compressed,kaloorazi2018subspace,martinsson2016randomized,rokhlin2010randomized,yu2018efficient}.

The rest of this paper is structured as follows. In section \ref{sec2}, we expound on some randomized algorithms and introduce some random test matrices for acceleration. Section \ref{sec3} demonstrates the theoretical analysis for a set of random test matrices and their corresponding algorithms for the fixed-precision problem. The error analysis is exhibited in section \ref{sec4}. We devise some numerical experiments in section \ref{sec5}, and deliver our conclusions along with avenues for future exploration in section \ref{sec6}.

\subsection{Notation}\label{sec1.1}
For the convenience of discussion, we assume throughout this paper that all matrices are real and let $\mathbb{R}^{m\times n}$ be the family of $m\times n$ real matrices. Consider a matrix $\mathbf{A}\in \mathbb{R}^{m\times n}$ ($m\geq n$), where its singular values $\{\sigma_{i}\}_{i=1}^{n}$ are arranged in non-increasing order with $\sigma_1\geq\cdots\geq\sigma_{n}\geq 0$. We use $\|\cdot\|$ and $\|\cdot\|_\mathrm{F}$ to depict the spectral and Frobenius norms of a specific matrix, respectively. We also use $\|\cdot\|$ to depict the Euclidean norm of a given vector. In statistical notation, we use $\mathbb{P}$ for probability, $\mathbb{E}$ for expectation, $\mbox{Var}$ for variance and $\mbox{Cov}$ for covariance, respectively.

\section{Technical preliminaries}\label{sec2}
This section provides an overview of some randomized algorithms that will be important throughout this paper. For an $m\times n$ matrix $\mathbf{A}$, one fundamental idea in randomized technique involves using random projections to approximate the principal subspace of $\mathbf{A}$ that contains its essential information. Specifically, the goal is to decompose $\mathbf{A}$ as $\mathbf{A}\approx\mathbf{Q}\mathbf{C},$ where $\mathbf{Q}\in\mathbb{R}^{m\times l}$ is a column orthonormal matrix, $\mathbf{C}=\mathbf{Q}^\top\mathbf{A}\in\mathbb{R}^{l\times n}$ is an upper triangular matrix, and $l$ is the desired rank or numerical rank of $\mathbf{A}$ to achieve the desired accuracy.

\subsection{Overview of the \texttt{farPCA} algorithm}\label{sec2.1}

A prototype randomized framework for constructing the approximation $\mathbf{A}\approx\mathbf{Q}\mathbf{Q}^\top \mathbf{A}=\mathbf{Q}\mathbf{C}$, called \texttt{randQB}, was conveyed in \cite[Algorithm 4.3]{halko2011finding}. In \texttt{randQB}, to seek the orthonormal basis for the range of $\mathbf{A}$ with rank $k$, randomized algorithms frequently capitalize on an $n\times (k+h)$ random test matrix $\mathbf{\Phi}$ and find the orthonormal basis for the range of $(\mathbf{A}\mathbf{A}^\top)^P\mathbf{A}\mathbf{\Phi}$. Typically, $\mathbf{\Phi}$ is chosen as a standard Gaussian matrix $\mathbf{G}$. Here the oversampling parameter $h$ is conventionally exploited to enhance performance while spanning the needed subspace. The power parameter $P$ is leveraged to augment performance, as it results in a sharper drop in the singular values of $(\mathbf{A}\mathbf{A}^\top)^P\mathbf{A}$ compared to those of $\mathbf{A}$. After producing the low-rank decomposition $\mathbf{A}\approx \mathbf{Q}\mathbf{C}$, the next is to construct a standard decomposition. For example, the direct SVD approach \cite[Algorithm 5.1]{halko2011finding} to form an approximate SVD $\mathbf{A}\approx\mathbf{U}\mathbf{\Sigma}\mathbf{V}^\top$ can be described as follows.

\begin{itemize}
    \item Step $1$: Compute an SVD of $\mathbf{C}$: $\mathbf{C}=\mathbf{\tilde{U}}\mathbf{\Sigma}\mathbf{V}^\top$.
    \item Step $2$: $\mathbf{U}=\mathbf{Q}\mathbf{\tilde{U}}$.
\end{itemize}

For the fixed-precision approximation problem, the approximation error is $\interleave(\mathbf{I}-\mathbf{Q}\mathbf{Q}^\top)\mathbf{A}\interleave$, where $\interleave\cdot\interleave$ denotes spectral norm or Frobenius norm. Then the objective of the fixed-precision approximation problem is to terminate the algorithm once the approximation error reaches a desired accuracy $\varepsilon$. In \cite{yu2018efficient}, the \texttt{randQB\_EI} algorithm was introduced, capitalizing on the blocked Gram-Schmidt procedure and an error indicator predicated on the assertion that $\|\mathbf{A}-\mathbf{Q}\mathbf{C}\|_\mathrm{F}^2=\|\mathbf{A}\|_\mathrm{F}^2-\|\mathbf{C}\|_\mathrm{F}^2$. This serves as a variation to \texttt{randQB}. The authors further substantiate that this indicator works well in double-precision floating arithmetic when the desired accuracy $\varepsilon>2.1\times 10^{-7}\|\mathbf{A}\|_\mathrm{F}$. The \texttt{randQB\_EI} algorithm with power iteration is detailed in \cite[Algorithm 2.1]{hallman2022block}.

In \cite{feng2023fast}, the \texttt{farPCA} algorithm was proposed by exploiting the shifted power iteration for superior accuracy and replacing QR decomposition in \texttt{randQB\_EI} with some matrix skills to augment parallel efficiency, as delineated in Algorithm \ref{alg1}. The \texttt{farPCA} is essentially equivalent to the integration of \texttt{randQB\_EI} and the direct SVD when the power parameter $P=0,1,2$, except that the singular values produced by Algorithm \ref{alg1} are sorted in ascending order while those from the direct SVD are in descending order. For $P>2$, the shifted parameter $\alpha$ is updated to enhance accuracy. Step 5 involves orthonormalizing $\mathbf{A}\mathbf{G}_j$ with respect to $\mathbf{Q}$, where $\mathbf{Q}$ is the already attained orthonormal basis for the range of $\mathbf{A}$. The $\texttt{eigSVD}$ algorithm from \cite{feng2023fast} is applied to compute the economic SVD efficiently.

\begin{algorithm}[htbp]
	\caption{\texttt{farPCA} with shifted power iteration \cite[Algorithm 5]{feng2023fast}}\label{alg1}
 \begin{algorithmic}[1]
	\REQUIRE{$\mathbf{A}\in \mathbb{R}^{m\times n}$, tolerance $\varepsilon$, power parameter $P$, block size $b$.}

\ENSURE{The SVD decomposition $\mathbf{A}\approx\mathbf{U}\mathbf{\Sigma}\mathbf{V}^\top$ satisfying $\|\mathbf{A}-\mathbf{U}\mathbf{\Sigma}\mathbf{V}^\top\|_\mathrm{F}\leq \varepsilon$.}

\STATE{$\mathbf{Y}=[~]$, $\mathbf{W}=[~]$, $\mathbf{Z}=[~]$. $E=\|\mathbf{A}\|_\mathrm{F}^2.$}

\FOR{$j=1,2,3,\cdots$}

\STATE{Generate an $n\times b$ standard Gaussian matrix $\mathbf{G}_j$, $\alpha=0$.}

\FOR{$k=1,2,\cdots,P$}

\STATE{$\mathbf{W}_j=\mathbf{A}^\top(\mathbf{A}\mathbf{G}_j)-\mathbf{W}\left(\mathbf{Z}^{-1}(\mathbf{W}^\top\mathbf{G}_j)\right)-\alpha\mathbf{G}_j$.}

\STATE{$[\mathbf{G}_j,\hat{\mathbf{\Sigma}},\sim]=\texttt{eigSVD}(\mathbf{W}_j)$}

\STATE{{\bf if} ($k>1$ and $\alpha<\hat{\mathbf{\Sigma}}(b,b)$) {\bf then } $\alpha={(\hat{\mathbf{\Sigma}}(b,b)+\alpha)}/{2}$.}

\ENDFOR

\STATE{$\mathbf{Y}_j=\mathbf{A}\mathbf{G}_j,\mathbf{W}_j=\mathbf{A}^\top\mathbf{Y}_j$, $\mathbf{Y}=[\mathbf{Y},\mathbf{Y}_{j}],\mathbf{W}=[\mathbf{W},\mathbf{W}_{j}]$.}


\STATE{$\mathbf{Z}=\mathbf{Y}^\top\mathbf{Y},\mathbf{T}=\mathbf{W}^\top\mathbf{W}$. $E=E-\mbox{tr}(\mathbf{T}\mathbf{Z}^{-1})$.}

\STATE{{\bf if} $E < \varepsilon^2$ {\bf then stop}.}

\ENDFOR

\STATE{Perform an eigenvalue decomposition: $[\hat{\mathbf{V}},\hat{\mathbf{S}}]=\text{eig}(\mathbf{Z})$. Compute $\mathbf{D}=\hat{\mathbf{V}}\hat{\mathbf{S}}^{-1/2}$.}

\STATE{Perform an eigenvalue decomposition: $[\tilde{\mathbf{V}},\tilde{\mathbf{S}}]=\text{eig}(\mathbf{D}^\top\mathbf{T}\mathbf{D})$. Compute $\mathbf{\Sigma}=\tilde{\mathbf{S}}^{1/2}$.}

\STATE{$\mathbf{U}=\mathbf{Y}\mathbf{D}\tilde{\mathbf{V}}$, $\mathbf{V}=\mathbf{W}\mathbf{D}\tilde{\mathbf{V}}\mathbf{\Sigma}^{-1}$.}

\end{algorithmic}
\end{algorithm}

Suppose that the multiplication of two dense matrices of sizes $m\times n$ and $n\times l$ incurs a computational cost of $C_{mm}mnl$ floating-point operations (flops), where $C_{mm}$ is a constant number, following the conventions used in \cite{martinsson2016randomized,yu2018efficient}. Assuming that the output $\mathbf{U}$ in Algorithm \ref{alg1} has $l$ columns, with $l$ being significantly smaller than $\min(m,n)$, the predominant cost of Algorithm \ref{alg1} is
\begin{equation}
    \begin{aligned}
        T_{Alg2.1}\approx 2C_{mm}mnl+\frac{3}{2}C_{mm}(m+n)l^2+P\left(C_{mm}\left(2mnl+nl^2+2nlb\right)\right).\notag
    \end{aligned}
\end{equation}

\subsection{Improvement of \texttt{farPCA}}

Instead of relying on the standard Gaussian matrix, some other random test matrices are utilized to accelerate computation, such as the Bernoulli matrix \cite{xu2023fast}, and the sparse sign matrix \cite{benjamin2017compressed,li2006very,martinsson2020randomized,tropp2019streaming,xie2023sketchne}. We describe the Bernoulli matrix and the sparse sign matrix in the following two definitions.
\begin{definition}[Bernoulli matrix \cite{huang2020rank,xu2023fast}]\label{def2.1}
    A matrix $\mathbf{B}\in\mathbb{R}^{n\times l}$ is called a Bernoulli matrix with parameter $p\in(0,1)$ if its entries $b_{ij}$ are independent and identically distributed ($i.i.d.$) drawn from Bernoulli distribution with parameter $p$, i.e. $\mathbb{P}(b_{ij}=1)=p$, $\mathbb{P}(b_{ij}=0)=1-p$.
\end{definition}

\begin{definition}[Sparse sign matrix \cite{benjamin2017compressed,li2006very}]\label{def2.2}
    A matrix $\mathbf{\Xi}\in\mathbb{R}^{n\times l}$ is called a sparse sign matrix with parameter $p\in(0,1)$ if its entries $\xi_{ij}$ are $i.i.d.$ drawn from the distribution: $\mathbb{P}(\xi_{ij}=1/\sqrt{p})=p/2$, $\mathbb{P}(\xi_{ij}=0)=1-p$, $\mathbb{P}(\xi_{ij}=-1/\sqrt{p})=p/2$.
\end{definition}

Another way to define the sparse sign matrix appears in \cite{xie2023sketchne,martinsson2020randomized,tropp2019streaming}, in which the crucial thought involves fixing the number of nonzero elements in each column. Overall, the two matrices serve as an alternative to the standard Gaussian matrix due to their sparsity. To enhance performance and provide guarantees for theoretical analysis, we now define the standardized Bernoulli matrix below. 

\begin{definition}[Standardized Bernoulli matrix]\label{def2.3}
Let $\mathbf{B}=(b_{ij})\in \mathbb{R}^{n\times l}$ be a Bernoulli matrix with parameter $p\in(0,1)$, and $\mathbf{\Omega}=(\omega_{ij})\in \mathbb{R}^{n\times l}$, where {\small$\omega_{ij}=(b_{ij}-p)/\sqrt{p(1-p)}$}. Then $\mathbf{\Omega}$ is a standardized Bernoulli matrix with parameter $p$.
\end{definition}

Notably, the standardized Bernoulli matrix includes the Rademacher (sign) matrix \cite{clarkson2009numerical} as a special case when $p=0.5$, which typically behaves like a standard Gaussian matrix \cite{tropp2017practical}. Additionally, the standardized Bernoulli matrix is a linear combination of the Bernoulli matrix and a rank-$1$ matrix, indicating that we can leverage the sparsity of the Bernoulli matrix to accelerate computation akin to the sparse sign matrix \cite{benjamin2017compressed,li2006very,martinsson2020randomized,tropp2019streaming,xie2023sketchne}. Furthermore, another matrix for better stability relative to the sparse sign matrix, termed sparse Gaussian matrix, is introduced below. Although mentioned in \cite{zhang2020faster}, it has not been applied to low-rank approximation. The sparse Gaussian matrix can capture more information than the sparse sign matrix, as its non-zero elements are drawn from the normal distribution, rather than being restricted to $\pm1/\sqrt{p}$.

\begin{definition}[Sparse Gaussian matrix \cite{zhang2020faster}]\label{def2.4}
    A matrix $\mathbf{\Psi}$ is called the sparse Gaussian matrix with parameter $p\in(0,1)$ if $\mathbf{\Psi}=\mathbf{B}\circ\mathbf{G}/\sqrt{p}$, where $\mathbf{B}$ is a Bernoulli matrix with parameter $p$, $\mathbf{G}$ is a standard Gaussian matrix, and $\circ$ denotes the hadamard product of matrices.
\end{definition}

\section{Low-rank matrix approximation based on sparse matrices}\label{sec3}

\subsection{Theoretical analysis for the standardized Bernoulli matrix and the \texttt{farPCASB} algorithm}\label{sec3.1}

In this section, our proposal is to create a standardized Bernoulli matrix with 
a small 
parameter $p$, leveraging the sparsity to accelerate computation while retaining the statistical characteristics of a standard Gaussian matrix. We analyze the convergence of the standardized Bernoulli matrix and subsequently develop the \texttt{farPCASB} algorithm. To prove the convergence of the standardized Bernoulli matrix, we will present the Lyapunov central limit theorem below (see \cite{billingsley2017probability,durrett2019probability}).

\begin{lemma}[Lyapunov central limit theorem]\label{lemma3.1}
    Let $\{X_1,\cdots,X_n,\cdots\}$ be a sequence of independent random variables. Assume the expected values $\mathbb{E}(X_i)=\mu_i$ and variances $\mbox{Var}(X_i)=\sigma_i^2$, where $\sigma_i^2$'s are finite for $i=1,2,\cdots,n,\cdots.$ Also let $s_n^2:=\sum_{k=1}^{n}\sigma_k^2$, if there exists $\delta>0$ such that the Lyapunov’s condition
    $$\lim_{n\to\infty}\frac{1}{s_n^{2+\delta}}\sum_{k=1}^{n}\mathbb{E}\left(\vert X_i-\mu_i\vert^{2+\delta}\right)=0$$
    is satisfied, then 
    $\sum_{k=1}^{n}(X_i-\mu_i)/{s_n}\mathop{\to}\limits^{d}\mathbb{N}(0,1)\text{ as }n\to\infty$,
    where $\mathop{\to}\limits^{d}$ indicates converge in distribution and $\mathbb{N}(0,1)$ signifies the standard normal distribution.
\end{lemma}

Given a matrix $\mathbf{A}\in\mathbb{R}^{m\times n}$ ($m\geq n$) with SVD: $\mathbf{A}=\mathbf{U}_A\mathbf{\Sigma}_A\mathbf{V}_A^\top$, where $\mathbf{U}_A\in\mathbb{R}^{m\times n},\mathbf{V}_A\in\mathbb{R}^{n\times n}$ are orthogonal matrices and $\mathbf{\Sigma}_A=\mbox{diag}(\sigma_1,\cdots,\sigma_{n})$ with $\sigma_1\geq\cdots\geq\sigma_{n}\geq 0$. Denote $\mathbf{\Omega}$ as the standardized Bernoulli matrix. The succeeding Theorem illustrates that as $n\to\infty$, $\mathbf{V}_A^\top\mathbf{\Omega}$ converges to a standard Gaussian matrix in distribution, assuming a mild condition is met. This guarantees that we can substitute $\mathbf{\Omega}$ for the standard Gaussian matrix in \texttt{farPCA}, which can leverage the merits of both the sparsity and the properties of the standard Gaussian matrix.

\begin{theorem}\label{thm3.2}
    Consider $\mathbf{\Omega}=(\omega_{ij})\in \mathbb{R}^{n\times l}$ as a standardized Bernoulli matrix with parameter $p\in(0,1)$. Let 
    $\mathbf{U}=(u_{ij})\in \mathbb{R}^{n\times n}$ be an orthogonal matrix satisfying the condition
    \begin{align}\label{eq3.1}
        \max\limits_{i,j=1,\cdots,n}\vert u_{ij}\vert\leq\varepsilon_n, \text{ where } \varepsilon_n\to0 \text{ as } n\to\infty,
    \end{align}
    then $\mathbf{U}\mathbf{\Omega}$ converges in distribution to a standard Gaussian matrix as $n\to\infty$. Furthermore, the Kolmogorov distance between the distribution of each element $(\mathbf{U}\mathbf{\Omega})_{ij}$ in $\mathbf{U}\mathbf{\Omega}$ and the standard normal distribution $\phi(x)$ is given by:
    \begin{align}\label{eq3.2}
        \sup\limits_{x}\vert \mathbb{P}((\mathbf{U}\mathbf{\Omega})_{ij}\leq x)-\phi(x)\vert=O\left(\sum_{k=1}^{n}|u_{ik}|^3\right).
    \end{align}
\end{theorem}

\begin{proof}[Proof of Theorem \ref{thm3.2}]
    Consider $\mathbf{B}=(b_{ij})\in\mathbb{R}^{n\times l}$ as a Bernoulli matrix with parameter $p\in(0,1)$. 
Then $\omega_{ij}$'s are $i.i.d.$  with $\mathbb{E}(\omega_{ij})=\mathbb{E}((b_{ij}-p)/\sqrt{p(1-p)})=0,\mbox{Var}(\omega_{ij})=\mbox{Var}(b_{ij})/(p(1-p))=1$. Let $\mathbf{\Omega}_1,\cdots,\mathbf{\Omega}_l$ represent the columns of $\mathbf{\Omega}$, and $\mathbf{U}^{(1)},\cdots,\mathbf{U}^{(n)}$ be the rows of $\mathbf{U}$, respectively. Thus the $(i,j)$th element of $\mathbf{U}\mathbf{\Omega}$ is $\mathbf{U}^{(i)}\mathbf{\Omega}_j=\sum_{k=1}^{n}u_{ik}\omega_{kj}.$
For each $k$, we have $\mathbb{E}(u_{ik}\omega_{kj})=u_{ik}\mathbb{E}(\omega_{kj})=0,\mbox{Var}(u_{ik}\omega_{kj})=u_{ik}^{2}\mbox{Var}(\omega_{kj})=u_{ik}^{2}.$
Let $X_k=u_{ik}\omega_{kj}$, then $\mathbb{E}(X_k)=0,\mbox{Var}(X_k)=u_{ik}^{2},s_n^2=\sum_{k=1}^{n}u_{ik}^{2}=1.$
Therefore, for any $\delta>0$,
    \begin{align}
        &\lim_{n\to\infty}\frac{1}{s_n^{2+\delta}}\sum_{k=1}^{n}\mathbb{E}\left(\vert X_k-\mu_k\vert^{2+\delta}\right)=\lim_{n\to\infty}\sum_{k=1}^{n}\mathbb{E}\left(\vert X_k\vert^{2+\delta}\right)\notag\\
        \leq&\lim_{n\to\infty}\sum_{k=1}^{n}\mathbb{E}\left(\vert X_k\vert^{2}\right)\left(\max_h\{|X_h|\}\right)^{\delta}     \leq\lim_{n\to\infty}\left(\varepsilon_n\max\left\{\sqrt{\frac{1-p}{p}},\sqrt{\frac{p}{1-p}}\right\}\right)^\delta=0,\notag
    \end{align}
because $\vert X_h\vert\in\{\sqrt{(1-p)/p}u_{ih},\sqrt{p/(1-p)}u_{ih}\}$ according to the definition of $\omega_{hj}$. Consequently, $\max\limits_h\{X_h\}\leq\max\limits_h\{|u_{ih}|\}\max\{\sqrt{(1-p)/p},\sqrt{p/(1-p)}\}$. As the maximum element value of $\mathbf{U}$ converges to $0$ as $n\to\infty$, the Lyapunov condition holds uniformly for all $i$ and $j$. By Lemma \ref{lemma3.1}, it follows that  $\mathbf{U}^{(i)}\mathbf{\Omega}_j\mathop{\to}\limits^{d}\mathbb{N}(0,1)\text{ as }n\to\infty$ for $i=1,\cdots,n;~j=1,\cdots,l$.

Now we demonstrate that the elements of $\mathbf{U}\mathbf{\Omega}$ are asymptotically independent as $n\to\infty$. 
Since  $\mathbb{E}(\mathbf{\Omega}_i)=\mathbf{0}_n$, $\mbox{Var}(\mathbf{\Omega}_i)=\mathbf{I}$, and $\mbox{Cov}(\mathbf{\Omega}_i,\mathbf{\Omega}_j)=\mathbb{E}(\mathbf{\Omega}_i\mathbf{\Omega}_j^\top)=\mathbf{0}_{n\times n}$ for $i,j=1,\cdots,l$ and $i\neq j$. Here $\mathbf{0}_n$ signifies an $n$-dimensional vector with all elements being $0$, while $\mathbf{0}_{n\times n}$ and $\mathbf{I}$ depict the $n\times n$ all-zero matrix and the $n\times n$ identity matrix, respectively. Considering the covariance between $(i,j)$th and $(k,l)$th element of $\mathbf{U}\mathbf{\Omega}$, where $(i,j)\neq(k,l)$, we have 
    \begin{align}
        \mbox{Cov}(\mathbf{U}^{(i)}\mathbf{\Omega}_j,\mathbf{U}^{(k)}\mathbf{\Omega}_l)&=\mathbb{E}(\mathbf{U}^{(i)}\mathbf{\Omega}_j \mathbf{\Omega}_l^\top \mathbf{U}_k^\top)=\mathbf{U}^{(i)}\mathbb{E}(\mathbf{\Omega}_j\mathbf{\Omega}_l^\top) {\mathbf{U}^{(k)}}^\top\notag\\
        &=\mathbf{U}^{(i)} \mathbf{0}_{n\times n} {\mathbf{U}^{(k)}}^\top=0,\text{ for }j\neq l,\text{ and }\notag\\
        \mbox{Cov}(\mathbf{U}^{(i)}\mathbf{\Omega}_j,\mathbf{U}^{(k)}\mathbf{\Omega}_l)&=\mathbf{U}^{(i)}\mathbb{E}(\mathbf{\Omega}_j\mathbf{\Omega}_l^\top) {\mathbf{U}^{(k)}}^\top\notag\\
        &=\mathbf{U}^{(i)} \mathbf{I} {\mathbf{U}^{(k)}}^\top=0\text{, for }j=l\text{ and }i\neq k.\notag
    \end{align}
The second expression above is valid due to the orthogonal property of matrix $\mathbf{U}$. Eventually, it can be indicated that all elements of $\mathbf{U}\mathbf{\Omega}$ are uncorrelated, which reveals they are also asymptotically independent in the case of a normal limiting distribution. 

Furthermore, according to the Berry-Esseen Theorem \cite[Theorem 12.4]{bhattacharya2010normal}, we have 
$$\sup\limits_{x}\vert \mathbb{P}((\mathbf{U}\mathbf{\Omega})_{ij}\leq x)-\phi(x)\vert=O\left(\frac{\sum_{k=1}^{n}\mathbb{E}|X_k|^3}{(s_n^2)^{3/2}}\right).$$ This implies that \eqref{eq3.2} holds as $\mathbb{E}(|X_k|^3)=|u_{ik}|^{3}\mathbb{E}(|\omega_{kj}|^3)$ and $\mathbb{E}(|\omega_{kj}|^3)$ is bounded.
\end{proof}

\begin{remark}\label{remark3.3}
      To implement Theorem \ref{thm3.2} to obtain the asymptotic distribution of $\mathbf{V}_A^\top\mathbf{\Omega}$, note that condition \eqref{eq3.1} is not strict and typically gets fulfilled for a general target matrix $\mathbf{A}\in\mathbb{R}^{m\times n}$. A similar condition appears in \cite{candes2012exact}, emphasizing that the singular vectors of the target matrix must be sufficiently distributed across all components to enable successful matrix completion from a sampled subset of entries. Some numerical experiments will be devised in subsection \ref{sec5.5} where the right singular matrix $\mathbf{V}_A$ of $\mathbf{A}$ is sparse, and the performance of \texttt{farPCASB} and some algorithms building on other random test matrices in this case will be recorded. Furthermore, \eqref{eq3.2} states the convergence rate of the distribution of elements in $\mathbf{U}\mathbf{\Omega}$, determined by the structure of $\mathbf{U}$. When each entry of $\mathbf{U}$ scales as $O(n^{-1/2})$, the rate of convergence is $n^{-1/2}$.
\end{remark}

By Theorem \ref{thm3.2}, $\mathbf{V}_A^\top\mathbf{\Omega}$ converges in distribution to a standard Gaussian matrix as $n\to\infty$ under condition \eqref{eq3.1}. This surmises that $\mathbf{A}\mathbf{\Omega}=\mathbf{U}_A\mathbf{\Sigma}_A\mathbf{V}_A^\top\mathbf{\Omega}$ asymptotically take the form of $\mathbf{A}\mathbf{G}=\mathbf{U}_A\mathbf{\Sigma}_A\mathbf{V}_A^\top\mathbf{G}$ for a standard Gaussian matrix $\mathbf{G}$, since the standard Gaussian matrix is rotationally invariant. It enables accelerated computation while asymptotically achieving the same results as the standard Gaussian matrix as 
$n\to\infty$. The \texttt{farPCASB} algorithm is exhibited in Algorithm \ref{alg2}.

\begin{algorithm}[htbp]
	\caption{\texttt{farPCA} based on the Standardized Bernoulli matrix (\texttt{farPCASB})}\label{alg2}
 \begin{algorithmic}[1]
	\REQUIRE{$\mathbf{A}\in \mathbb{R}^{m\times n}$, tolerance $\varepsilon$, power parameter $P$, block size $b$, parameter $p$ of Bernoulli matrix.}

\ENSURE{The SVD decomposition $\mathbf{A}\approx\mathbf{U}\mathbf{\Sigma}\mathbf{V}^\top$ satisfying $\|\mathbf{A}-\mathbf{U}\mathbf{\Sigma}\mathbf{V}^\top\|_\mathrm{F}\leq \varepsilon$.}

\STATE{$\mathbf{Y}=[~]$, $\mathbf{W}=[~]$, $\mathbf{Z}=[~]$. $E=\|\mathbf{A}\|_\mathrm{F}^2.$}

\STATE{$\mathbf{w}=p\mathbf{1}_n, \mathbf{z}=\mathbf{A}\mathbf{w}$, where $\mathbf{1}_n$ is an $n$-dimensional all-ones vector.}

\FOR{$j=1,2,3,\cdots$}

\STATE{Generate an $n\times b$ Bernoulli matrix $\mathbf{B}_j$ with parameter $p$. $\alpha=0$.}

\STATE{$\hat{\mathbf{B}}_j=\text{sparse}(\mathbf{B}_j)$.}

\IF{$P=0$}

\STATE{$\mathbf{Y}_j=\mathbf{A}\hat{\mathbf{B}}_j-\mathbf{z},\mathbf{W}_j=\mathbf{A}^\top\mathbf{Y}_j$.}

\ELSIF{$P\geq 1$}

\STATE{$\mathbf{W}_j=\mathbf{A}^\top(\mathbf{A}\hat{\mathbf{B}}_j-\mathbf{z})-\mathbf{W}\left(\mathbf{Z}^{-1}(\mathbf{W}^\top\hat{\mathbf{B}}_j-\mathbf{W}^\top\mathbf{w})\right)$.}

\STATE{$[\hat{\mathbf{B}}_j,\sim,\sim]=\text{eigSVD}(\mathbf{W}_j)$}

\FOR{$k=2,\cdots,P$}

\STATE{$\mathbf{W}_j=\mathbf{A}^\top(\mathbf{A}\hat{\mathbf{B}}_j)-\mathbf{W}\left(\mathbf{Z}^{-1}(\mathbf{W}^\top\hat{\mathbf{B}}_j)\right)-\alpha\hat{\mathbf{B}}_j$.}

\STATE{$[\hat{\mathbf{B}}_j,\hat{\mathbf{\Sigma}},\sim]=\text{eigSVD}(\mathbf{W}_j)$.}

\STATE{{\bf if} ($\alpha<\hat{\mathbf{\Sigma}}(b,b)$) {\bf then } $\alpha={(\hat{\mathbf{\Sigma}}(b,b)+\alpha)}/{2}$.}

\ENDFOR

\STATE{$\mathbf{Y}_j=\mathbf{A}\hat{\mathbf{B}}_j,\mathbf{W}_j=\mathbf{A}^\top\mathbf{Y}_j$.}

\ENDIF

\STATE{$\mathbf{Y}=[\mathbf{Y},\mathbf{Y}_{j}],\mathbf{W}=[\mathbf{W},\mathbf{W}_{j}]$, $\mathbf{Z}=\mathbf{Y}^\top\mathbf{Y},\mathbf{T}=\mathbf{W}^\top\mathbf{W}$. $E=E-\mbox{tr}(\mathbf{T}\mathbf{Z}^{-1})$.}


\STATE{{\bf if} $E < \varepsilon^2$ {\bf then stop}.}

\ENDFOR

\STATE{Compute an eigenvalue decomposition: $[\hat{\mathbf{V}},\hat{\mathbf{S}}]=\text{eig}(\mathbf{Z})$, and $\mathbf{D}=\hat{\mathbf{V}}\hat{\mathbf{S}}^{-1/2}$.}

\STATE{Compute an eigenvalue decomposition: $[\tilde{\mathbf{V}},\tilde{\mathbf{S}}]=\text{eig}(\mathbf{D}^\top\mathbf{T}\mathbf{D})$, and $\mathbf{\Sigma}=\tilde{\mathbf{S}}^{1/2}$.}

\STATE{$\mathbf{U}=\mathbf{Y}\mathbf{D}\tilde{\mathbf{V}}$, $\mathbf{V}=\mathbf{W}\mathbf{D}\tilde{\mathbf{V}}\mathbf{\Sigma}^{-1}$.}
\end{algorithmic}
\end{algorithm}

It is observed that the constant $1/\sqrt{p(1-p)}$ in the $n\times b$ standardized Bernoulli matrix $\mathbf{\Omega}_j$ can be eliminated, yielding the equation $\sqrt{p(1-p)}\mathbf{A}\mathbf{\Omega}_j=\mathbf{A}\mathbf{B}_j-p\mathbf{A}\mathbf{1}_n$, where the symbol ``$-$" represents matrix-vector subtraction in MATLAB. Accordingly, for different $j$, we employ $\mathbf{A}\mathbf{B}_j-\mathbf{z}$ and $\mathbf{W}^\top\mathbf{B}_j-\mathbf{W}^\top\mathbf{w}$ to replace $\mathbf{A}\mathbf{\Omega}_j$ and $\mathbf{W}^\top\mathbf{\Omega}_j$, respectively, where $\mathbf{z}=\mathbf{A}\mathbf{w}$ and $\mathbf{w}=p\mathbf{1}_n$.

We now analyze the computational expenditure of \texttt{farPCASB} in comparison to \texttt{farPCA}. For \texttt{farPCASB}, computational savings primarily arise from the acceleration techniques, while minor additional costs stem from calculations of $\mathbf{w},\mathbf{z},\mathbf{W}^\top\mathbf{w}$, along with a few extra additions. These additional costs are negligible despite MATLAB’s inefficiency with additional subtraction operations in Step 7 or Step 9. Thus, we directly concentrate on $\mathbf{A}\mathbf{B}_j$, where the cost term $C_{mm}mnl$ in $T_{Alg2.1}$ (for $\mathbf{A}\mathbf{G}_j$) can be replaced by $\sum_{j=1}^{l/b}\mbox{nnz}(\mathbf{B}_j)m=\mbox{nnz}(\mathbf{B})m$, with $\mbox{nnz}(\cdot)$ depicting the number of nonzero elements of a given matrix and $\mathbf{B}$ being an $n\times l$ Bernoulli matrix. To be more precise, by Hoeffding’s inequality for bounded random variables \cite[Proposition 2.5]{wainwright2019high}, we have $\mathbb{P}(|\mbox{nnz}(\mathbf{B})-nlp|\geq t)\leq\exp\left(-2t^2/(nl)\right)$ for all $t\geq0$, which speculates that computing $\mathbf{A}\mathbf{B}$ is more efficient than $\mathbf{A}\mathbf{G}$ when $p$ is small enough. Therefore, 
\[
  T_{Alg3.1}\!\approx \!C_{mm}mnl\!+\!mnlp\!+\!O\!\!\left(\!m\sqrt{nl}\!\right)\!+\!\frac{3}{2}C_{mm}(m\!+\!n)l^2\!+\!P\!\left(C_{mm}\!\left(2mnl\!+\!nl^2\!+\!2nlb\!\right)\!\right)
\]

\begin{remark}\label{remark3.4}
    When $l$ and the parameter $p$ is small enough, it can be seen that $T_{Alg3.1}/T_{Alg2.1}\approx(1+2P)C_{mm}mnl/\left((2+2P)C_{mm}mnl\right)= (1+2P)/(2+2P)$. Therefore, \texttt{farPCASB} can serve as a superior alternative to \texttt{farPCA}, 
   offering almost identical performance while achieving a computational cost reduction of $(100/(2+2P))\%$.
\end{remark}

\subsection{Low-rank approximation by other sparse matrices}
As described in \cite{martinsson2020randomized,xie2023sketchne}, the sparse sign matrix also behaves like a standard Gaussian matrix. In fact, it can be understood that the sparse sign matrix exhibits a property analogous to that described in Theorem \ref{thm3.2} for the standardized Bernoulli matrix. We summarize it in Proposition \ref{proposition3.5}.

\begin{proposition}\label{proposition3.5}
    Consider $\mathbf{\Xi}\in \mathbb{R}^{n\times l}$ as a sparse sign matrix with parameter $p\in(0,1)$. Under condition \eqref{eq3.1} on the $n\times n$ orthogonal matrix 
 $\mathbf{U}$, $\mathbf{U}\mathbf{\Xi}$ converges in distribution to a standard Gaussian matrix as $n\to\infty$.
\end{proposition}

Likewise, the sparse Gaussian matrix exhibits a property analogous to the sparse sign matrix and the standardized Bernoulli matrix. It is described in Proposition \ref{proposition3.6}.

\begin{proposition}\label{proposition3.6}
    Consider $\mathbf{\Psi}\in \mathbb{R}^{n\times l}$ as a sparse Gaussian matrix with parameter $p$, where $p\in(0,1)$. Under condition \eqref{eq3.1} on the $n\times n$ orthogonal matrix 
 $\mathbf{U}$, $\mathbf{U}\mathbf{\Psi}$ converges in distribution to a standard Gaussian matrix as $n\to\infty$.
\end{proposition}
\begin{remark}\label{remark3.7}
    The conclusion is not confined to these specific matrix types. In fact, it is applicable to any matrix in which each element is $i.i.d.$ with mean $0$ and variance $1$. The convergence rates of the distribution function of each element are determined by the structure of $\mathbf{U}$. We omit the detailed proof of Propositions \ref{proposition3.5} and \ref{proposition3.6} as the reasoning follows the identical steps outlined in the proof of Theorem \ref{thm3.2}.
\end{remark}

Notably, both the sparse sign matrix $\mathbf{\Xi}$ and the sparse Gaussian matrix $\mathbf{\Psi}$ satisfy Hoeffding's inequality for bounded random variables: $$\mathbb{P}(|\mbox{nnz}(\mathbf{\Xi})-nlp|\geq t)\leq\exp\left(-\frac{2t^2}{nl}\right),\mathbb{P}(|\mbox{nnz}(\mathbf{\Psi})-nlp|\geq t)\leq\exp\left(-\frac{2t^2}{nl}\right)$$ for all $t\geq0$. Accordingly, substituting the standard Gaussian matrix in \texttt{farPCA} with either $\mathbf{\Xi}$ (\texttt{farPCASS}) or $\mathbf{\Psi}$ (\texttt{farPCASG}) maintains computational costs roughly equivalent to $T_{Alg3.1}$. Furthermore, \texttt{farPCASS} and \texttt{farPCASG} may surpass \texttt{farPCASB} in speed, as the latter involves additional subtraction operations in Steps 7 or 9.

\section{Error analysis}\label{sec4}
In this section, we perform an error analysis on our proposed algorithms: \texttt{farPCASB}, \texttt{farPCASS}, and \texttt{farPCASG}. Error analysis is typically performed on the fixed-rank problem. Starting with an asymptotic error analysis of \texttt{randQB}, we derive tighter error bounds than those in \cite{halko2011finding}. Since \texttt{randQB} is mathematically equivalent to \texttt{farPCA} without shifted technique, this analysis extends to \texttt{farPCA} as well. Building on our results in Theorem \ref{thm3.2} and Propositions \ref{proposition3.5} and \ref{proposition3.6}, we extend the error analysis to all three algorithms, excluding shifted power iteration.

Now we give some notation necessarily in this section. Recall that for a given matrix $\mathbf{A}\in\mathbb{R}^{m\times n}$ ($m\geq n$) with its SVD: $\mathbf{A}=\mathbf{U}_A\mathbf{\Sigma}_A\mathbf{V}_A^\top$. Our current goal is to analyze the approximation's error constructed by projecting $\mathbf{A}$ to the orthonormal basis $\mathbf{Q}$, which is the subspace spanned by the range of $\mathbf{A}\mathbf{\Phi}$. Here $\mathbf{\Phi}\in\mathbb{R}^{n\times l}$ is a chosen random matrix, with $l=k+h$, $k$ being a desired rank and $h$ being an oversampling parameter. The SVD of $\mathbf{A}$ is now rewritten as
\begin{align}\label{eq4.1}
\mathbf{A}=\mathbf{U}_A\begin{bmatrix}\mathbf{\Sigma}_1 & \\& \mathbf{\Sigma}_2\end{bmatrix}
\begin{bmatrix}\mathbf{V}_1^\top\\
    \mathbf{V}_2^\top\end{bmatrix},
\end{align}
where $\mathbf{\Sigma}_1\in\mathbb{R}^{k\times k},\mathbf{\Sigma}_2\in\mathbb{R}^{(n-k)\times (n-k)}$ are diagonal square matrices and $\mathbf{V}_1^\top\in\mathbb{R}^{k\times n},\mathbf{V}_2^\top\in\mathbb{R}^{(n-k)\times n}$. Define 
\begin{align}\label{eq4.2}
\mathbf{\Phi}_1=\mathbf{V}_1^\top\mathbf{\Phi},\mathbf{\Phi}_2=\mathbf{V}_2^\top\mathbf{\Phi},\end{align}
then the lemma below states the error bound for any chosen $\mathbf{\Phi}$ \cite{halko2011finding}.

\begin{lemma}[{\cite[Theorem 9.1]{halko2011finding}}]\label{lemma4.1}
    Construct $\mathbf{Y}=\mathbf{A}\mathbf{\Phi}$ for any chosen $\mathbf{\Phi}$, and continue using the notation from equation \eqref{eq4.1} and \eqref{eq4.2}. Consequently, suppose that $\mathbf{\Phi}_1$ has full row rank and $\mathbf{Q}$ is an orthonormal basis for the range of $\mathbf{Y}$, then
    \begin{align}\label{eq4.3}
            \interleave(\mathbf{I}-\mathbf{Q}\mathbf{Q}^\top)\mathbf{A}\interleave^2\leq\interleave\mathbf{\Sigma}_2\interleave^2+\interleave\mathbf{\Sigma}_2\mathbf{\Phi}_2\mathbf{\Phi}_1^\dagger\interleave^2,
    \end{align}
    where $\dagger$ depicts the pseudo-inverse and $\interleave\cdot\interleave$ signifies the spectral or Frobenius norm.
\end{lemma}

Below are several lemmas utilized to derive the approximation error's probabilistic and expected bounds. The following two lemmas present some statistical instruments required for our later analysis \cite[Proposition 10.1, Proposition 10.3]{halko2011finding}.

\begin{lemma}[{\cite[Proposition 10.1]{halko2011finding}}]\label{lemma4.2}
For any fixed real matrices $\mathbf{S},\mathbf{T}$, let $\mathbf{G}\in\mathbb{R}^{n\times l}$ be a standard Gaussian matrix, then
    \begin{align}
        &\left(\mathbb{E}\|\mathbf{SGT}\|_\mathrm{F}^2\right)^{\frac{1}{2}}=\|\mathbf{S}\|_\mathrm{F}\|\mathbf{T}\|_\mathrm{F} \text{ and}\label{eq4.4}\\
        &\mathbb{E}\|\mathbf{SGT}\|\leq\|\mathbf{S}\|\|\mathbf{T}\|_\mathrm{F}+\|\mathbf{S}\|_\mathrm{F}\|\mathbf{T}\|.\label{eq4.5}
    \end{align}
\end{lemma}

\begin{lemma}[{\cite[Proposition 10.3]{halko2011finding}}]\label{lemma4.3}
Let $f$ be a function on matrices that is Lipschitz with constant $L$, i.e., $|f(\mathbf{X})-f(\mathbf{Y})|\leq L\|\mathbf{X}-\mathbf{Y}\|_\mathrm{F}$ for any matrices $\mathbf{X},\mathbf{Y}$, then for a standard Gaussian matrix $\mathbf{G}$ and all $t>0$, 
\begin{equation}
    \begin{aligned}
        \mathbb{P}\left(f(\mathbf{G})\geq \mathbb{E}f(\mathbf{G})+Lt\right)\leq e^{-t^2/2}.\notag
    \end{aligned}
\end{equation}
\end{lemma}

To gain the expected and probabilistic error bounds, we shall render a few studies in RMT to control the norm of the pseudo-inverse of a standard Gaussian matrix more precisely in the asymptotic setting, instead of utilizing bounds supplied by \cite[Proposition 10.2, Proposition 10.4]{halko2011finding}. The ensuing definition provides the spectral distribution of a matrix, a fundamental instrument for illustrating spectral traits in RMT.

\begin{definition}\label{def4.4}
    The empirical spectral distribution (ESD) of an $m\times m$ matrix $\mathbf{A}$ with real eigenvalues is described as follows: 
    $$F^\mathbf{A}(x)=\frac{1}{m}\sum_{j=1}^{m}\mathbb{I}(\lambda_j(\mathbf{A})\leq x),$$
    where $\lambda_j(\mathbf{A})$ is the $j$th largest eigenvalue of $\mathbf{A}$ and $\mathbb{I}(\cdot)$ is the indicator function.
\end{definition}

Some analyse exists for ascertaining the asymptotic behavior of the ESD for a specified series of random matrices. The following lemma elucidates the limiting spectral behavior of $\mathbf{G}\mathbf{G}^\top$ \cite[Theorem 3.6]{bai2010spectral}, where $\mathbf{G}$ is a standard Gaussian matrix.

\begin{lemma}[M-P law for the sample covariance matrices]\label{lemma4.5}
    Suppose that $\mathbf{X}=(x_{ij})\in \mathbb{R}^{m\times n}$ be $i.i.d.$ real random variables with mean $0$ and variance $\sigma^2$, and define $\mathbf{S}=\mathbf{X}\mathbf{X}^\top/n$. Furthermore, assume that $m/n\to \gamma\in(0,+\infty)$ as $n,p\to\infty$, then with probability $1$, the ESD of $\mathbf{S}$ converges weakly to the M-P law $F_\gamma$, where $F_\gamma$ is a cumulative distribution function with its probability density function as follows:
    \begin{equation}\label{eq4.6}
        \begin{aligned}
            f_\gamma(x)=\left\{\begin{array}{ll}\displaystyle\frac{1}{2\pi x\gamma\sigma^2}\sqrt{(b-x)(x-a)},&\text{if}\quad a\leq x\leq b,\\0,&\text{otherwise,}\end{array}\right.
        \end{aligned}
    \end{equation}
    where $a=\sigma^2(1-\sqrt{\gamma})^2,b=\sigma^2(1+\sqrt{\gamma})^2$, $F_\gamma$ has a point mass $1-1/{\gamma}$ at $0$ if $\gamma>1$.
\end{lemma}

By Lemma \ref{lemma4.5}, more exact norm bounds for the pseudo-inverse of the Gaussian matrix are delineated below in the asymptotic setting. 

\begin{lemma}\label{lemma4.6}
    For an $m\times n$ $(m\leq n)$ standard Gaussian matrix $\mathbf{G}$, let $m/n\to\gamma\in(0,1)$ as $n\to\infty$, then with probability $1$,
    \begin{align}
        &\|\mathbf{G}^\dagger\|_\mathrm{F}\to\sqrt{\frac{\gamma}{1-\gamma}} \text{ and }
       \lim\limits_{n\to\infty}\|\mathbf{G}^\dagger\|\leq\sqrt{\left(n(1-\sqrt{\gamma})^2\right)^{-1}}.\label{eq4.7}
    \end{align}
\end{lemma}

\begin{proof}[Proof of Lemma \ref{lemma4.6}]
    By Lemma \ref{lemma4.5}, the ESD of $\mathbf{G}\mathbf{G}^\top/n$ converges weakly to $F_\gamma$ with $\sigma^2=1$, where $F_\gamma$ has a density function $f_\gamma$ as delineated in \eqref{eq4.6}. The support of $f_\gamma$ is $[(1-\sqrt{\gamma})^2,(1+\sqrt{\gamma})^2]$ with probability $1$. 
    Then $$\|\mathbf{G}^\dagger\|^2=\lambda_1\left((\mathbf{G}\mathbf{G}^\top)^{-1}\right)=\left(\lambda_m(\mathbf{G}\mathbf{G}^\top)\right)^{-1}=\left(n\lambda_m(\mathbf{G}\mathbf{G}^\top/n)\right)^{-1}\leq\left(n(1-\sqrt{\gamma})^2\right)^{-1},$$ where $\lambda_j(\cdot)$ is the $j$th largest eigenvalue of the specified matrix. For the Frobenius norm, we have $\|\mathbf{G}^\dagger\|_\mathrm{F}^2=\mbox{tr}\left((\mathbf{G}\mathbf{G}^\top)^{-1}\right)=m/n\cdot\mbox{tr}\left((\mathbf{G}\mathbf{G}^\top/n)^{-1}\right)/m.$
    Further we are going to prove that $$\frac{1}{m}\mbox{tr}\left((\mathbf{G}\mathbf{G}^\top/n)^{-1}\right)=\int \frac{1}{x}\mathrm{d}F^{\mathbf{G}\mathbf{G}^\top/n}(x)\to \int \frac{1}{x}\mathrm{d}F_\gamma(x),$$ where $F^{\mathbf{G}\mathbf{G}^\top/n}(x)$ is the ESD of $\mathbf{G}\mathbf{G}^\top/n$.
    In fact, the convergence of the integral above holds since the weak convergence of the cumulative distribution function $F^{\mathbf{G}\mathbf{G}^\top/n}$ by Lemma \ref{lemma4.5}. Specifically, given the weak convergence of the cumulative distribution function, one can infer the weak convergence of the Lebesgue-Stieltjes measure defined by the corresponding cumulative distribution function. In the light of Lemma \ref{lemma4.5}, the smallest eigenvalue of $\mathbf{G}\mathbf{G}^\top/n$ is larger than $(1-\sqrt{\gamma})^2$ almost surely, then $(1-\sqrt{\gamma})^{-2}\wedge (1/x)$ is continuous and bounded, where $\wedge$ denotes the minimum value of two numbers. This extrapolates the convergence of the integral above in line with the Portmanteau Theorem.

 For $\gamma<1$, letting $x=1+\gamma+2\sqrt{\gamma}\cos w$, where $w\in[0,\pi]$, and $\zeta=e^{iw}$, then
    \begin{small}
        \begin{align*}
            &\int \frac{1}{x}\mathrm{d}F_\gamma(x)=\int_{0}^{\pi}\frac{2}{\pi}\frac{\sin^{2}w}{(1+\gamma+2\sqrt{\gamma}\cos w)^{2}}\mathrm{d}w=\frac{1}{\pi}\int_{0}^{2\pi}\frac{((e^{iw}-e^{-iw})/2i)^{2}}{(1+\gamma+\sqrt{\gamma}(e^{iw}+e^{-iw}))^{2}}\mathrm{d}w\\
   &=-\frac{1}{4\pi i}\int_{|\zeta|=1}\frac{(\zeta^2-1)^{2}}{\zeta\gamma\left((\zeta+\sqrt{\gamma})(\zeta+1/\sqrt{\gamma})\right)^{2}}\mathrm{d}\zeta=\frac{1}{1-\gamma}.\notag
\end{align*}
\end{small}
The integrand above is calculated by Cauchy integration with 
 three poles at $\zeta_0=0,\zeta_1=-\sqrt{\gamma},\zeta_2=-1/\sqrt{\gamma}$ when $\gamma<1$, where $\zeta_0$ is a simple pole and the order of $\zeta_2,\zeta_3$ is $2$. This extrapolates that $\|\mathbf{G}^\dagger\|_\mathrm{F}^2\to{\gamma}/{(1-\gamma)}$.
 \end{proof}


The succeeding lemma expounds the expected and probabilistic error bounds of \texttt{randQB} \cite{halko2011finding}, which are exactly the same as those of \texttt{farPCA} without shifted technique, owing to their mathematical equivalence.

\begin{lemma}[Expected and probabilistic error bounds for \texttt{randQB} {\cite{halko2011finding}}]\label{lemma4.7}
    Generate an $n\times l$ standard Gaussian matrix $\mathbf{G}$, where $l=k+h\leq n$, $k\geq 2$, and $h\geq 2$. Construct $\mathbf{Y}=\mathbf{A}\mathbf{G}$. Consequently, if $\mathbf{Q}$ is an orthonormal basis for the range of $\mathbf{Y}$ and $\mathbf{Q}_P$ serves as an orthonormal basis for the range of $(\mathbf{A}\mathbf{A}^\top)^P\mathbf{Y}$ with $P$ being a natural number, then
    \begin{footnotesize}
    \begin{align}
        &\mathbb{E}\|(\mathbf{I}-\mathbf{Q}\mathbf{Q}^\top)\mathbf{A}\|_\mathrm{F}\leq\left(1+\frac{k}{h-1}\right)^{\frac{1}{2}}\left(\sum_{j=k+1}^{n}\sigma_j^2\right)^{\frac{1}{2}},\label{eq4.9}\\
        &\mathbb{E}\|(\mathbf{I}-\mathbf{Q}_P\mathbf{Q}_P^\top)\mathbf{A}\|\leq\left(\left(1+\sqrt{\frac{k}{h-1}}\right)\sigma_{k+1}^{2P+1}+\frac{e\sqrt{k+h}}{h}\left(\sum_{j=k+1}^{n}\sigma_j^{2(2P+1)}\right)^{\frac{1}{2}}\right)^{\frac{1}{2P+1}}.\label{eq4.10}
    \end{align}
    \end{footnotesize}
    Further assume that $h\geq 4$, then for all $u,t\geq 1$,
    \begin{small}
    \begin{align}
        \|(\mathbf{I}-\mathbf{Q}\mathbf{Q}^\top)\mathbf{A}\|_\mathrm{F}\leq\left(1+t\sqrt{\frac{3k}{h+1}}\right)\left(\sum_{j=k+1}^{n}\sigma_j^2\right)^{\frac{1}{2}}+ut\frac{e\sqrt{k+h}}{h+1}\sigma_{k+1}\label{eq4.11}
    \end{align}
    \end{small}
    with failure probability not exceeding $2t^{-h}+e^{-u^2/2}$, and 
    \begin{footnotesize}
    \begin{multline}\label{eq4.12}
    \|(\mathbf{I}-\mathbf{Q}_P\mathbf{Q}_P^\top)\mathbf{A}\|\leq\\
    \left(\left(1+t\sqrt{\frac{3k}{h+1}}\right)\sigma_{k+1}^{2P+1}+t\frac{e\sqrt{k+h}}{h+1}\left(\sum_{j=k+1}^{n}\sigma_j^{2(2P+1)}\right)^{\frac{1}{2}}+ut\frac{e\sqrt{k+h}}{h+1}\sigma_{k+1}^{2P+1}\right)^{\frac{1}{2P+1}}
    \end{multline}
    \end{footnotesize}
     with failure probability not exceeding $2t^{-h}+e^{-u^2/2}$.
\end{lemma}

To enable effective error control for large matrices, we leverage innovations from RMT to establish tighter asymptotic bounds. Instead of employing the norm bounds of the pseudo-inverse of the Gaussian matrix presented by \cite[Proposition 10.2, Proposition 10.4]{halko2011finding}, Lemma \ref{lemma4.6} is used to give a more exact norm bound when $k$ and $h$ tend to infinity. The results are as follows. 

\begin{theorem}[Asymptotic expected and probabilistic error bounds for \texttt{randQB} (\texttt{farPCA} without shifted power iteration)]\label{thm4.8}
    Generate an $n\times l$ standard Gaussian matrix $\mathbf{G}$, where $l=k+h\leq n$, $k\geq 2$, and $h\geq 2$. Let $\mathbf{Y}=\mathbf{A}\mathbf{G}$, then if $\mathbf{Q}$ is an orthonormal basis for the range of $\mathbf{Y}$ and $\mathbf{Q}_P$ serves as an orthonormal basis for the range of $(\mathbf{A}\mathbf{A}^\top)^P\mathbf{Y}$ with $P$ being a natural number, as $k\to\infty,k/(k+h)\to\gamma\in(0,1)$, the subsequent equations hold with probability $1$:
    \begin{footnotesize}
    \begin{align}
        &\mathbb{E}\|(\mathbf{I}-\mathbf{Q}\mathbf{Q}^\top)\mathbf{A}\|_\mathrm{F}\leq\left(1+\frac{k}{h}\right)^{\frac{1}{2}}\left(\sum_{j=k+1}^{n}\sigma_j^2\right)^{\frac{1}{2}},\label{eq4.13}\\
        &\mathbb{E}\|(\mathbf{I}-\mathbf{Q}_P\mathbf{Q}_P^\top)\mathbf{A}\|\leq\left(\left(1+\sqrt{\frac{k}{h}}\right)\sigma_{k+1}^{2P+1}+\frac{\sqrt{k+h}+\sqrt{k}}{h}\left(\sum_{j=k+1}^{n}\sigma_j^{2(2P+1)}\right)^{\frac{1}{2}}\right)^{\frac{1}{2P+1}}.\label{eq4.14}
    \end{align}
    \end{footnotesize}
    Furthermore, for all $u>0$,
    \begin{small}
    \begin{align}
        \|(\mathbf{I}-\mathbf{Q}\mathbf{Q}^\top)\mathbf{A}\|_\mathrm{F}\leq\left(1+\sqrt{\frac{k}{h}}\right)\left(\sum_{j=k+1}^{n}\sigma_j^2\right)^{\frac{1}{2}}+u\frac{\sqrt{k+h}+\sqrt{k}}{h}\sigma_{k+1},\label{eq4.15}
    \end{align}
    \end{small}
    with failure probability not exceeding $e^{-u^2/2}$, and 
    \begin{footnotesize}
    \begin{multline}\label{eq4.16}
    \|(\mathbf{I}-\mathbf{Q}_P\mathbf{Q}_P^\top)\mathbf{A}\|\leq\\
    \left(\left(1+\sqrt{\frac{k}{h}}\right)\sigma_{k+1}^{2P+1}+\frac{\sqrt{k+h}+\sqrt{k}}{h}\left(\sum_{j=k+1}^{n}\sigma_j^{2(2P+1)}\right)^{\frac{1}{2}}+u\frac{\sqrt{k+h}+\sqrt{k}}{h}\sigma_{k+1}^{2P+1}\right)^{\frac{1}{2P+1}}
    \end{multline}
    \end{footnotesize}
     with failure probability not exceeding $e^{-u^2/2}$.
\end{theorem}

\begin{proof}[Proof of Theorem \ref{thm4.8}]
    We restrict our proof to the non-power iteration case, with the bounds for the power iteration version being straightforwardly inferred from \cite[Proposition 8.6]{halko2011finding}. According to Lemma \ref{lemma4.1}, when the random matrix $\mathbf{\Phi}$ is chosen as a standard Gaussian matrix $\mathbf{G}$, denote $\mathbf{\Phi}_1=\mathbf{V}_1^\top\mathbf{G},\mathbf{\Phi}_2=\mathbf{V}_2^\top\mathbf{G}$, then
    $$\interleave(\mathbf{I}-\mathbf{Q}\mathbf{Q}^\top)\mathbf{A}\interleave^2\leq\interleave\mathbf{\Sigma}_2\interleave^2+\interleave\mathbf{\Sigma}_2\mathbf{\Phi}_2\mathbf{\Phi}_1^\dagger\interleave^2.$$
    Therefore, according to the H{\"o}lder's inequality, $$\mathbb{E}\|(\mathbf{I}-\mathbf{Q}\mathbf{Q}^\top)\mathbf{A}\|_\mathrm{F}\leq\left(\mathbb{E}\|(\mathbf{I}-\mathbf{Q}\mathbf{Q}^\top)\mathbf{A}\|_\mathrm{F}^2\right)^{\frac{1}{2}}\leq\left(\|\mathbf{\Sigma}_2\|_\mathrm{F}^2+\mathbb{E}\|\mathbf{\Sigma}_2\mathbf{\Phi}_2\mathbf{\Phi}_1^\dagger\|_\mathrm{F}^2\right)^{\frac{1}{2}},$$
    then by applying \eqref{eq4.4}, \eqref{eq4.7}, and the law of total expectation, we have $$\mathbb{E}\|\mathbf{\Sigma}_2\mathbf{\Phi}_2\mathbf{\Phi}_1^\dagger\|_\mathrm{F}^2=\mathbb{E}\left(\mathbb{E}\left(\left\|\mathbf{\Sigma}_2\mathbf{\Phi}_2\mathbf{\Phi}_1^\dagger\right\|_\mathrm{F}^2\mid\mathbf{\Phi}_1\right)\right)=\mathbb{E}\left(\left\|\mathbf{\Sigma}_2\right\|_\mathrm{F}^2\left\|\mathbf{\Phi}_1^\dagger\right\|_\mathrm{F}^2\right)=\frac{k}{h}\|\mathbf{\Sigma}_2\|_\mathrm{F}^2.$$
    The last equality above holds since $\|\mathbf{G}^\dagger\|_\mathrm{F}\to\sqrt{{k}/{h}}$, implying $\mathbb{E}\|\mathbf{G}^\dagger\|_\mathrm{F}^2\to k/h$. Consequently, $\mathbb{E}\|(\mathbf{I}-\mathbf{Q}\mathbf{Q}^\top)\mathbf{A}\|_\mathrm{F}\leq\left(1+k/h\right)^{1/2}\left(\sum_{j=k+1}^{n}\sigma_j^2\right)^{1/2}$, thereby \eqref{eq4.13} holds.
    
    For the expected spectral error bound, by H{\"o}lder's inequality and \eqref{eq4.5}, then
    \begin{equation}\begin{aligned}
\mathbb{E}\|\mathbf{\Sigma}_{2}\mathbf{\Phi}_2\mathbf{\Phi}_1^{\dagger}\|\leq\mathbb{E}(\|\mathbf{\Sigma}_{2}\|\|\mathbf{\Phi}_1^{\dagger}\|_{\mathrm{F}}+\|\mathbf{\Sigma}_{2}\|_{\mathrm{F}}\|\mathbf{\Phi}_1^{\dagger}\|).\notag
\end{aligned}\end{equation}
According to Lemma \ref{lemma4.6}, $\lim\limits_{k\to\infty}\|\mathbf{G}^\dagger\|\leq(\sqrt{k+h}-\sqrt{k})^{-1}$ . Therefore,
\begin{small}
\begin{align*}
        \mathbb{E}\|(\mathbf{I}-\mathbf{Q}\mathbf{Q}^\top)\mathbf{A}\|&\leq\left(\mathbb{E}\|(\mathbf{I}-\mathbf{Q}\mathbf{Q}^\top)\mathbf{A}\|^2\right)^{\frac{1}{2}}\leq\left(\|\mathbf{\Sigma}_2\|^2+\mathbb{E}\|\mathbf{\Sigma}_2\mathbf{\Phi}_2\mathbf{\Phi}_1^\dagger\|^2\right)^{\frac{1}{2}}\\
        &\leq\|\mathbf{\Sigma}_2\|+\mathbb{E}\|\mathbf{\Sigma}_2\mathbf{\Phi}_2\mathbf{\Phi}_1^\dagger\|\leq\|\mathbf{\Sigma}_2\|+\left\|\mathbf{\Sigma}_2\right\|\mathbb{E}\|\mathbf{\Phi}_1^\dagger\|_\mathrm{F}+\left\|\mathbf{\Sigma}_2\right\|_\mathrm{F}\cdot\mathbb{E}\|\mathbf{\Phi}_1^\dagger\|\\
        &\leq\left(1+\sqrt{\frac{k}{h}}\right)\sigma_{k+1}+\frac{\sqrt{k+h}+\sqrt{k}}{h}\left(\sum_{j=k+1}^{n}\sigma_j^2\right)^{\frac{1}{2}}.
\end{align*}
\end{small}

For the probabilistic error bounds, consider the function $f_1(\mathbf{X})=\|\mathbf{\Sigma}_2\mathbf{X}\mathbf{\Phi}_1^\dagger\|,\\f_2(\mathbf{X})=\|\mathbf{\Sigma}_2\mathbf{X}\mathbf{\Phi}_1^\dagger\|_\mathrm{F}$, then $|f_1(\mathbf{X})-f_1(\mathbf{Y})|\leq\|\mathbf{\Sigma}_2\|\|\mathbf{X}-\mathbf{Y}\|\|\mathbf{\Phi}_1^\dagger\|\leq\|\mathbf{\Sigma}_2\|\|\mathbf{X}-\mathbf{Y}\|_\mathrm{F}\|\mathbf{\Phi}_1^\dagger\|$, which implies that $f_1$ is $L$-Lipschitz, where $L\leq\|\mathbf{\Sigma}_2\|\|\mathbf{\Phi}_1^\dagger\|.$ Moreover, \eqref{eq4.5} ascertains that $\mathbb{E}\left(f_1(\mathbf{\Phi}_2)|\mathbf{\Phi}_1\right)\leq\|\mathbf{\Sigma}_2\|\|\mathbf{\Phi}_1^\dagger\|_\mathrm{F}+\|\mathbf{\Sigma}_2\|_\mathrm{F}\|\mathbf{\Phi}_1^\dagger\|$. Relying on Lemma \ref{lemma4.3} intimates that $$\mathbb{P}\left(\|\mathbf{\Sigma}_2\mathbf{\Phi}_2\mathbf{\Phi}_1^\dagger\|\geq\|\mathbf{\Sigma}_2\|\|\mathbf{\Phi}_1^\dagger\|_{\mathrm{F}}+\|\mathbf{\Sigma}_2\|_{\mathrm{F}}\|\mathbf{\Phi}_1^\dagger\|+\|\mathbf{\Sigma}_2\|\|\mathbf{\Phi}_1^\dagger\|\cdot u\right)\leq\mathrm{e}^{-u^2/2}.$$

According to Lemma \ref{lemma4.6}, we have
{\small$$\mathbb{P}\left(\|\mathbf{\Sigma}_2\mathbf{\Phi}_2\mathbf{\Phi}_1^\dagger\|\geq\|\mathbf{\Sigma}_2\|\sqrt{\frac{k}{h}}+\|\mathbf{\Sigma}_2\|_{\mathrm{F}}\frac{\sqrt{k+h}+\sqrt{k}}{h}+\|\mathbf{\Sigma}_2\|\frac{\sqrt{k+h}+\sqrt{k}}{h}\cdot u\right)\leq\mathrm{e}^{-u^2/2}.$$}
Then \eqref{eq4.16} holds since {\small$\|(\mathbf{I}-\mathbf{Q}\mathbf{Q}^\top)\mathbf{A}\|\leq\left(\|\mathbf{\Sigma}_2\|^2+\|\mathbf{\Sigma}_2\mathbf{\Phi}_2\mathbf{\Phi}_1^\dagger\|^2\right)^{1/2}\leq\|\mathbf{\Sigma}_2\|+\|\mathbf{\Sigma}_2\mathbf{\Phi}_2\mathbf{\Phi}_1^\dagger\|$.}

For the Frobenius case, note that $|f_2(\mathbf{X})-f_2(\mathbf{Y})|\leq\|\mathbf{\Sigma}_2\|\|\mathbf{X}-\mathbf{Y}\|_\mathrm{F}\|\mathbf{\Phi}_1^\dagger\|$, which intimates that $f_2$ is also $L$-Lipschitz, where $L\leq\|\mathbf{\Sigma}_2\|\|\mathbf{\Phi}_1^\dagger\|.$ Then by H{\"o}lder's inequality and \eqref{eq4.4}, $\mathbb{E}\left(f_2(\mathbf{\Phi}_2)|\mathbf{\Phi}_1\right)\leq\left(\mathbb{E}\Big(\|\mathbf{\Sigma}_2\mathbf{\Phi}_2\mathbf{\Phi}_1^\dagger\|_\mathrm{F}^2|\mathbf{\Phi}_1\Big)\right)^{\frac{1}{2}}=\|\mathbf{\Sigma}_2\|_\mathrm{F}\|\mathbf{\Phi}_1^\dagger\|_\mathrm{F}.$ Therefore, Lemma \ref{lemma4.3} extrapolates that $$\mathbb{P}(\|\mathbf{\Sigma}_2\mathbf{\Phi}_2\mathbf{\Phi}_1^\dagger\|_\mathrm{F}\geq\|\mathbf{\Sigma}_2\|_\mathrm{F}\|\mathbf{\Phi}_1^\dagger\|_\mathrm{F}+\|\mathbf{\Sigma}_2\|\|\mathbf{\Phi}_1^\dagger\| u)\leq\mathrm{e}^{-u^2/2}.$$
By Lemma \ref{lemma4.6}, \eqref{eq4.15} holds since $\|(\mathbf{I}-\mathbf{Q}\mathbf{Q}^\top)\mathbf{A}\|_\mathrm{F}\leq\left(\|\mathbf{\Sigma}_2\|_\mathrm{F}^2+\|\mathbf{\Sigma}_2\mathbf{\Phi}_2\mathbf{\Phi}_1^\dagger\|_\mathrm{F}^2\right)^{1/2}\leq\|\mathbf{\Sigma}_2\|_\mathrm{F}+\|\mathbf{\Sigma}_2\mathbf{\Phi}_2\mathbf{\Phi}_1^\dagger\|_\mathrm{F}$.
\end{proof}

Moreover, relying on Theorem \ref{thm3.2} and Propositions \ref{proposition3.5} and \ref{proposition3.6}, we further expand these error bounds for \texttt{farPCASB}, \texttt{farPCASS} and \texttt{farPCASG} by exploiting Lemma \ref{lemma4.7} and Theorem \ref{thm4.8}. We emphasize the error analysis of algorithms that do not utilize shifted power iteration, as the power parameters $P=0,1,2$ generally suffice for most applications. The shifted power iteration only works when $P>2$ and yields nearly optimal approximation results for larger values of $P$ (e.g. $P=5$ \cite{feng2023fast}).

\begin{theorem}[Error bounds for \texttt{farPCASB}, \texttt{farPCASS} and \texttt{farPCASG} without shifted power iteration]\label{thm4.9}
    Assume that $\mathbf{V}_A^\top=(u_{ij})$ satisfy condition \eqref{eq3.1}. Generate a standardized Bernoulli matrix or a sparse sign matrix or a sparse Gaussian matrix $\mathbf{\Upsilon}\in\mathbb{R}^{n\times l}$ with parameter $p\in(0,1)$, where $l=k+h\leq n$, $k\geq 2$ and $h\geq 2$. Consequently, if $\mathbf{Q}$ is an orthonormal basis for the range of $\mathbf{Y}=\mathbf{A}\mathbf{\Upsilon}$ and $\mathbf{Q}_P$ serves as an orthonormal basis for the range of $(\mathbf{A}\mathbf{A}^\top)^P\mathbf{Y}$ with $P$ being a natural number, then with probability $1$, \eqref{eq4.9} and \eqref{eq4.10} hold. Further suppose that $h\geq 4$, for all $u,t\geq 1$, \eqref{eq4.11} and \eqref{eq4.12} hold with failure probability not exceeding $2t^{-h}+e^{-u^2/2}$.

    Additionally, as $k\to\infty,k/(k+h)\to\gamma\in(0,1)$, \eqref{eq4.13} and \eqref{eq4.14} hold with probability $1$, and \eqref{eq4.15} and \eqref{eq4.16} hold with failure probability not exceeding $e^{-u^2/2}$ for all $u>0$.
\end{theorem}

\begin{proof}
    For \eqref{eq4.9}-\eqref{eq4.12}, the proof straightforwardly derives from Lemma \ref{lemma4.7} since $\mathbf{V}_A^\top\mathbf{\Upsilon}$ is asymptotically equivalent to $\mathbf{V}_A^\top\mathbf{G}$ by exploiting Theorem \ref{thm3.2} and Propositions \ref{proposition3.5} and \ref{proposition3.6}. For the asymptotic error bounds \eqref{eq4.13}-\eqref{eq4.16}, the proof follows directly from Theorem \ref{thm4.8}, with the same argument that $\mathbf{V}_A^\top\mathbf{\Upsilon}$ is asymptotically equivalent to $\mathbf{V}_A^\top\mathbf{G}$.
\end{proof}

\section{Numerical experiments}\label{sec5}
In this section, we present some numerical experiments to elucidate the effectiveness of \texttt{farPCASB}, \texttt{farPCASS}, and \texttt{farPCASG} relative to \texttt{farPCA} and \texttt{farPCA} based on the Bernoulli matrix (\texttt{farPCAB}). For the real data, MATLAB's \texttt{svds} function is exploited to produce precise results, after which \texttt{farPCASB}, \texttt{farPCASS}, and \texttt{farPCASG} are compared with \texttt{farPCA}, \texttt{farPCAB}, \texttt{randUBV}, and \texttt{svds} to assess its computational performance. Results for addressing the fixed-rank approximation and fixed-precision approximation problems are delineated in subsections \ref{sec5.1} and \ref{sec5.2}. Furthermore, the analysis of how variations in the parameter 
$p$ affect the performance of the four acceleration algorithms is presented in subsection \ref{sec5.3}. We also examine these algorithms on real data in subsection \ref{sec5.4} and finally discuss their stability in subsection \ref{sec5.5}.

For simulations, we generate two kinds of $n\times n$ matrices:
\begin{itemize}
    \item Matrix $1$ (slow decay): Let $\mathbf{A}$ be the form $\mathbf{A}=\mathbf{U}_A\mathbf{\Sigma}_A\mathbf{V}_A^\top$, where $\mathbf{U}_A,\mathbf{V}_A\in\mathbb{R}^{n\times n}$ are attained by orthonormalizing the standard Gaussian matrices. The diagonal values of the diagonal matrix $\mathbf{\Sigma}_A$ are: $\sigma_j=1/j^2$ for $1\leq j\leq n$. This matrix is Matrix $1$ in \cite[Section 5.1]{hallman2022block}.
    \item Matrix $2$ (fast decay): $\mathbf{A}$ is built in the same manner as Matrix $2$, except that the singular values become: $\sigma_j=e^{-j/20}$ for $1\leq j\leq n$. This matrix is Matrix 3 in \cite[Section 5.1]{hallman2022block}.
\end{itemize}

Execution time, rank, and error metrics for the algorithms in subsections \ref{sec5.1} to \ref{sec5.3} are averaged over multiple trials with identical input matrices and $i.i.d.$ random test matrices. Denoting the approximation of $\mathbf{A}$ by $\hat{\mathbf{A}}$, the relative Frobenius-norm error is $\|\mathbf{A}-\hat{\mathbf{A}}\|_\mathrm{F}/\|\mathbf{A}\|_\mathrm{F}.$ All experiments were conducted in MATLAB 2024b on an Intel Core i9-12900KF CPU with 64 GB of DDR5 RAM.

\subsection{Efficiency validation for solving the fixed-rank problems}\label{sec5.1}

In practical applications of fixed-rank problems, rank $l$ is selected directly as the number of columns of $\mathbf{U}$ output in these algorithms. Notably, an exceedingly small $p$ could result in rank deficiency, requiring more iterations. Specifically, when $p$ is slightly larger, the algorithm stabilizes and becomes less sensitive to the choice of 
$p$, as will be elaborated in subsection \ref{sec5.3}. Nevertheless, $p$ must be minimal enough to capitalize on the sparsity (e.g., $p\leq0.1$). Practically, for an $m\times n$ target matrix with $m\geq n$, we recommend setting $p=10^{-3}$ when $n$ is large, which can leverage Theorem \ref{thm3.2} to attain the asymptotic behavior of $\mathbf{V}_A^\top\mathbf{B}$. Conversely, for smaller $n$, as the asymptotic results do not apply, a relatively larger value of $p$ is recommended to refrain from the risk of rank deficiency. Therefore, we empirically set $p_1=\max\{10^{-3},\ln(n)/n\}$ for \texttt{farPCASB} and $p_2=\max\{10^{-3},10/n\}$ for \texttt{farPCASS}, \texttt{farPCASG}, and \texttt{farPCAB}. To evaluate the performance, two types of input matrices of size $n\times n$ were generated, with $n$ set to $5000$, $10000$, and $30000$. The block size $b$ fixed at $20$, while $l$ ranges from $b$ to $10b$. The average results of $20$ experiments for power parameter $P=1$ and $P=0$ are summarized in Figure \ref{fig1} and Figure \ref{fig2}, respectively.

\begin{figure}[htbp]
    \begin{subfigure}{\textwidth}
	\centering
	\begin{minipage}{0.4\textwidth}
			\centering
			\includegraphics[width=\textwidth]{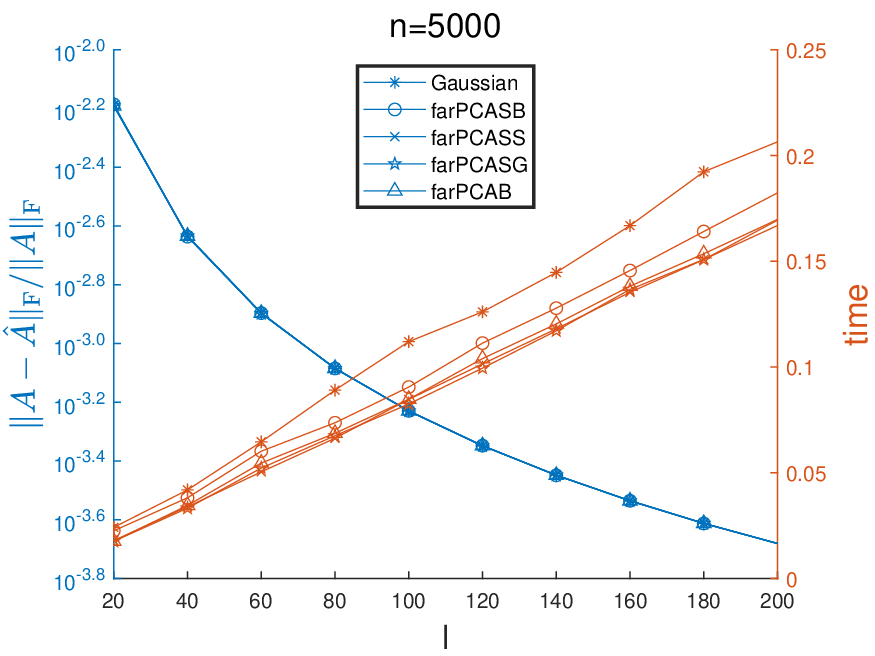}
		\end{minipage}
		\begin{minipage}{0.4\textwidth}
			\centering
			\includegraphics[width=\textwidth]{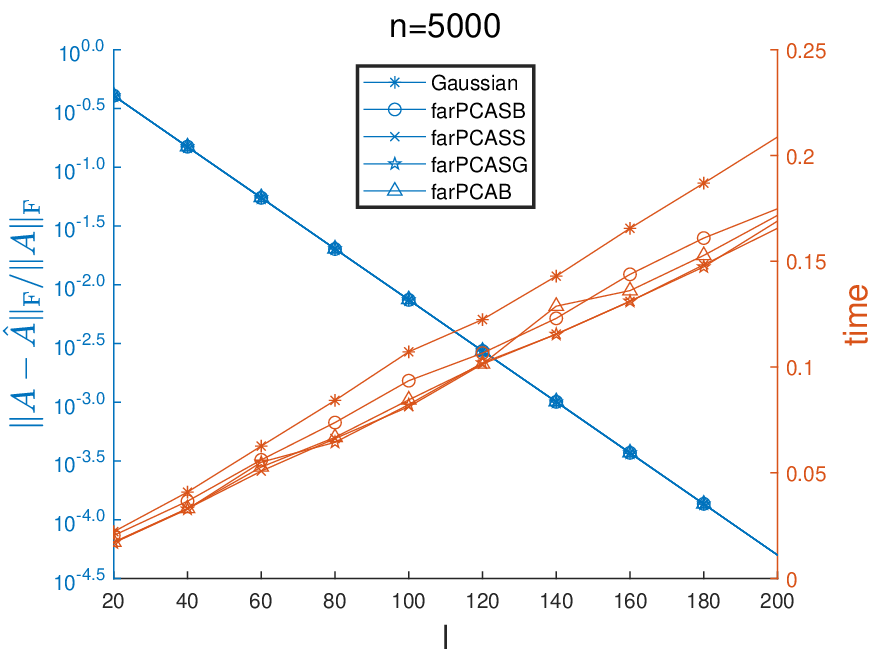}
	\end{minipage}
    \end{subfigure}
 
     \begin{subfigure}{\textwidth}
	\centering
	\begin{minipage}{0.4\textwidth}
			\centering
			\includegraphics[width=\textwidth]{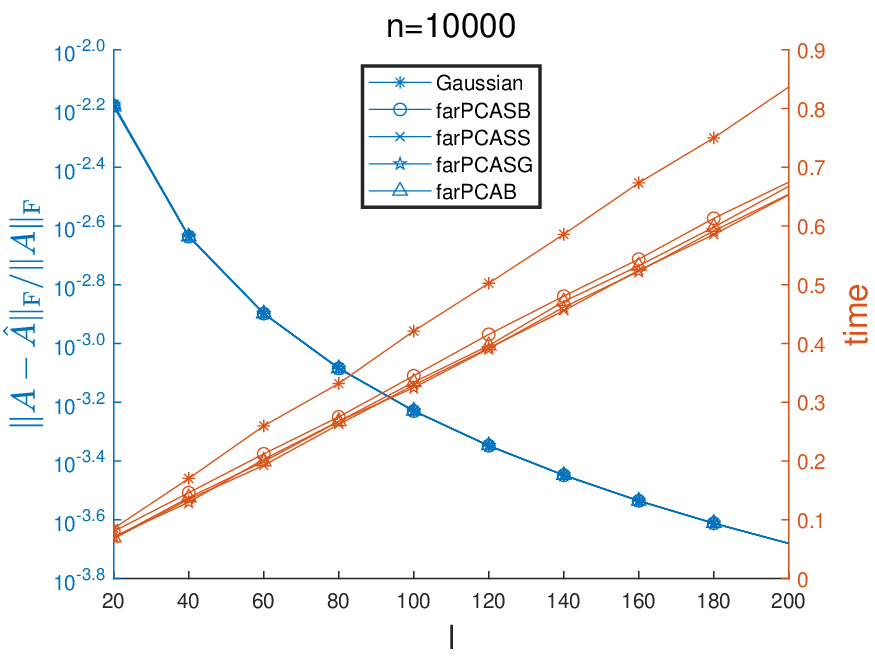}
		\end{minipage}
		\begin{minipage}{0.4\textwidth}
			\centering
			\includegraphics[width=\textwidth]{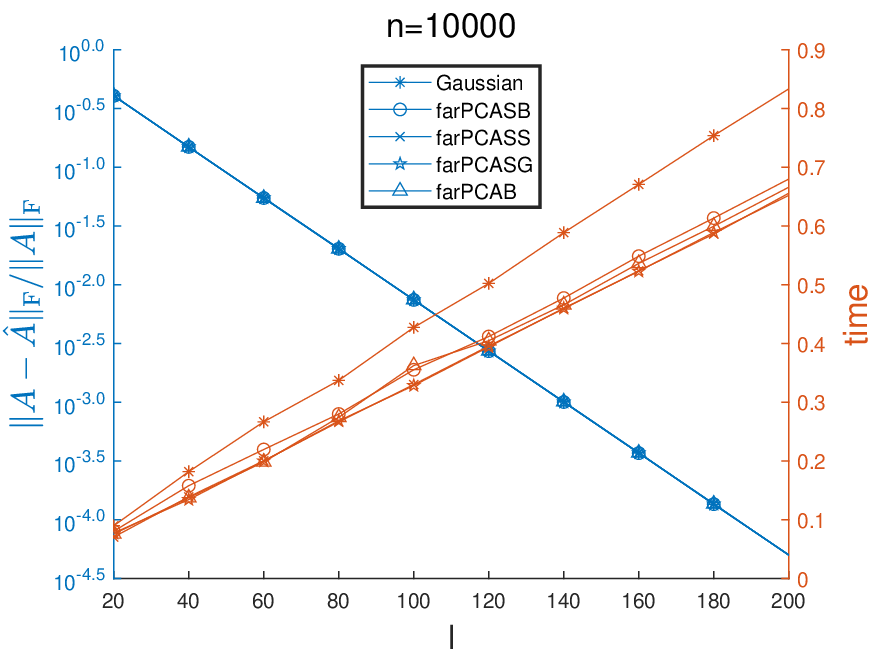}
	\end{minipage}
    \end{subfigure}

    \begin{subfigure}{\textwidth}
	\centering
	\begin{minipage}{0.4\textwidth}
			\centering
			\includegraphics[width=\textwidth]{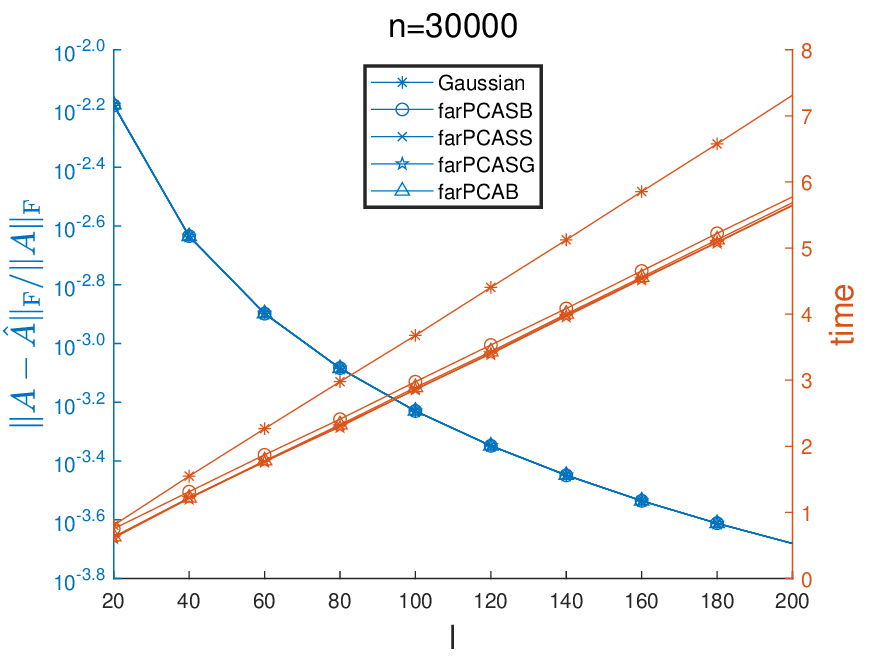}
		\end{minipage}
		\begin{minipage}{0.4\textwidth}
			\centering
			\includegraphics[width=\textwidth]{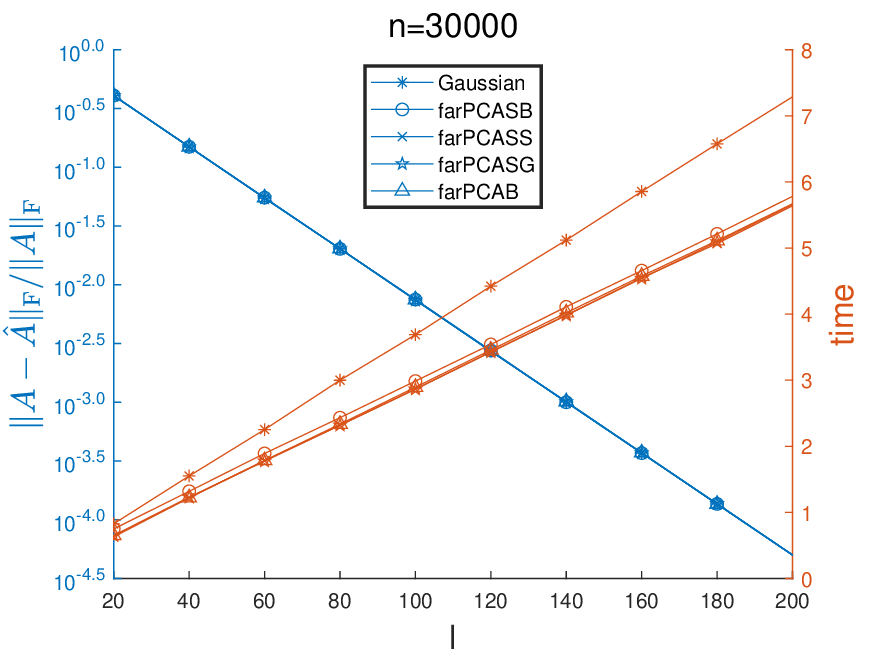}
	\end{minipage}
    \end{subfigure}
   	\caption{Comparison of relative Frobenius-norm errors and execution times between Gaussian (\texttt{farPCA}), \texttt{farPCASB}, \texttt{farPCASS}, \texttt{farPCASG} and \texttt{farPCAB} with varying value of $l$ under power parameter $P=1$ (averaged over $20$ trials). The first two columns correspond to Matrix $1$ and Matrix $2$, and rows denote matrices of varying dimensions.}
   \label{fig1} 
\end{figure}

\begin{figure}[htbp]
    \begin{subfigure}{\textwidth}
	\centering
	\begin{minipage}{0.4\textwidth}
			\centering
			\includegraphics[width=\textwidth]{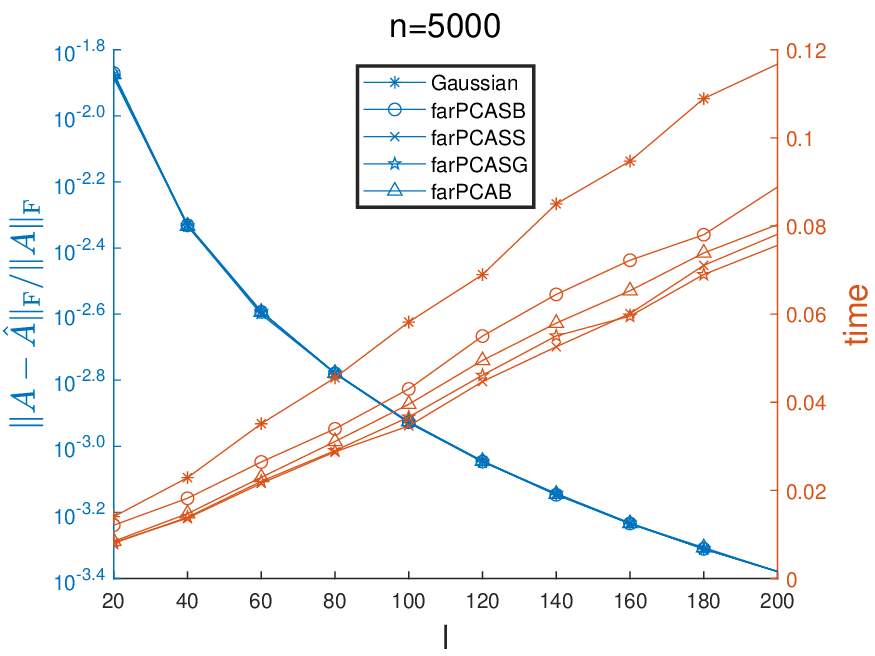}
		\end{minipage}
		\begin{minipage}{0.4\textwidth}
			\centering
			\includegraphics[width=\textwidth]{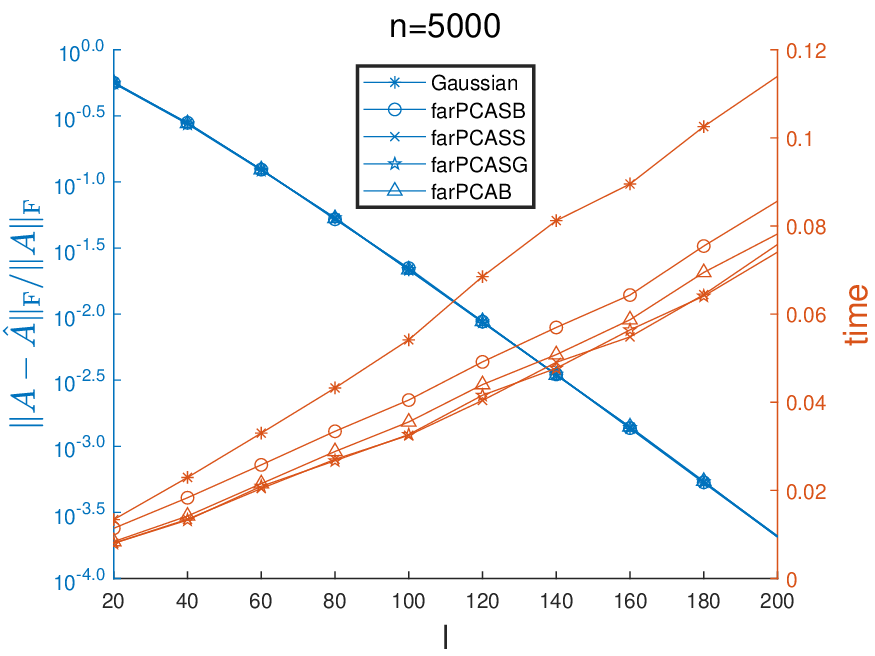}
	\end{minipage}
    \end{subfigure}
 
     \begin{subfigure}{\textwidth}
	\centering
	\begin{minipage}{0.4\textwidth}
			\centering
			\includegraphics[width=\textwidth]{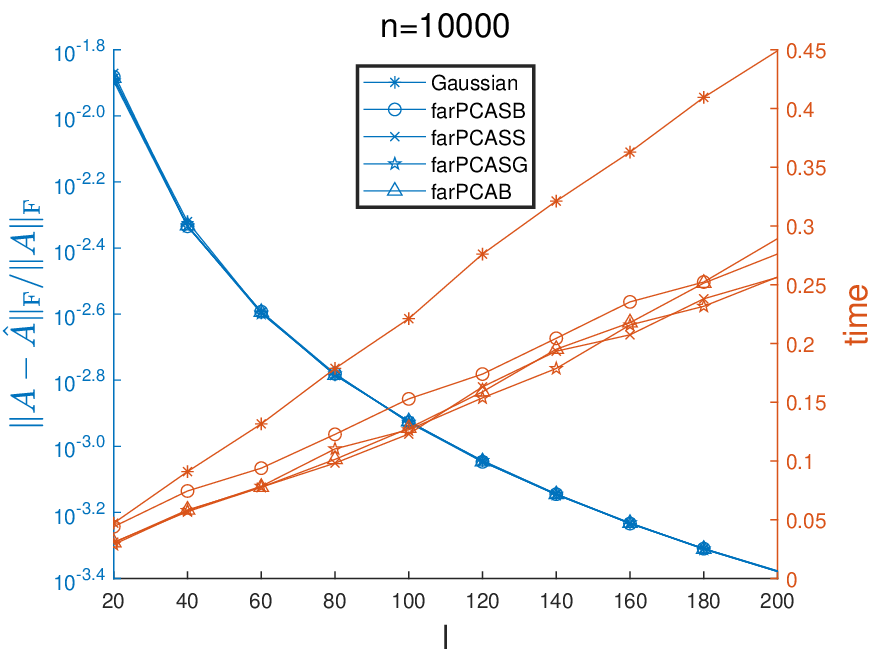}
		\end{minipage}
		\begin{minipage}{0.4\textwidth}
			\centering
			\includegraphics[width=\textwidth]{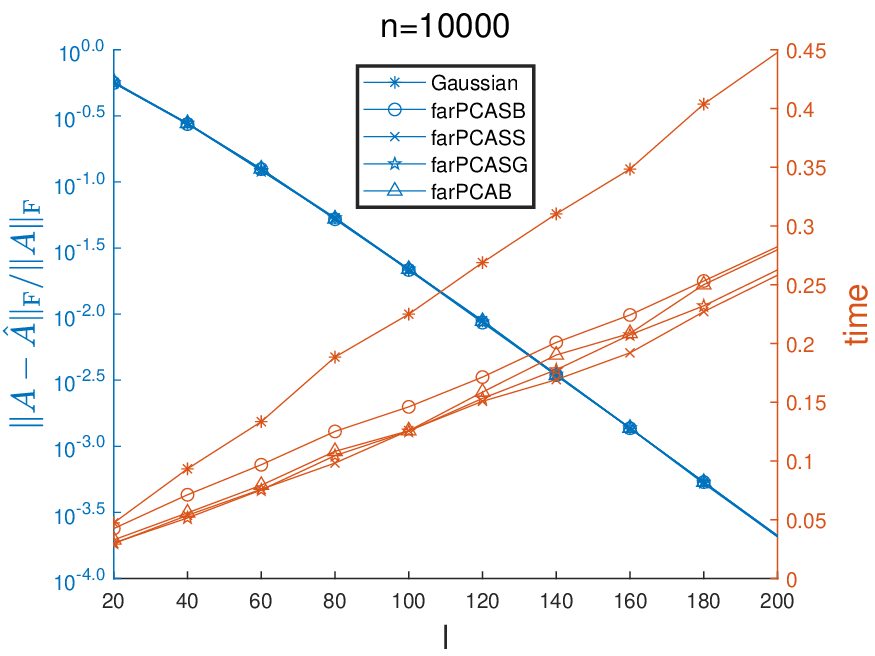}
	\end{minipage}
    \end{subfigure}

    \begin{subfigure}{\textwidth}
	\centering
	\begin{minipage}{0.4\textwidth}
			\centering
			\includegraphics[width=\textwidth]{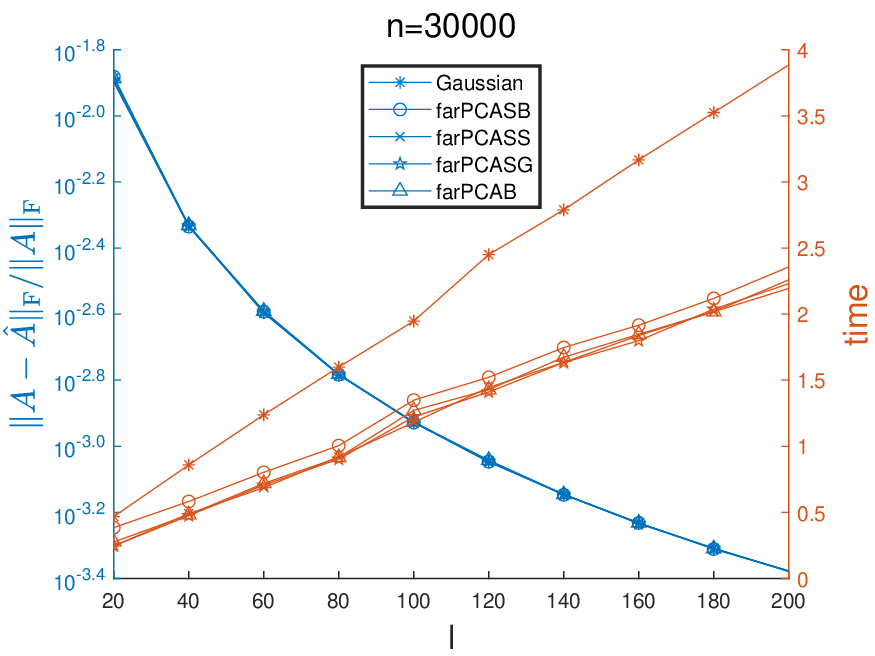}
		\end{minipage}
		\begin{minipage}{0.4\textwidth}
			\centering
			\includegraphics[width=\textwidth]{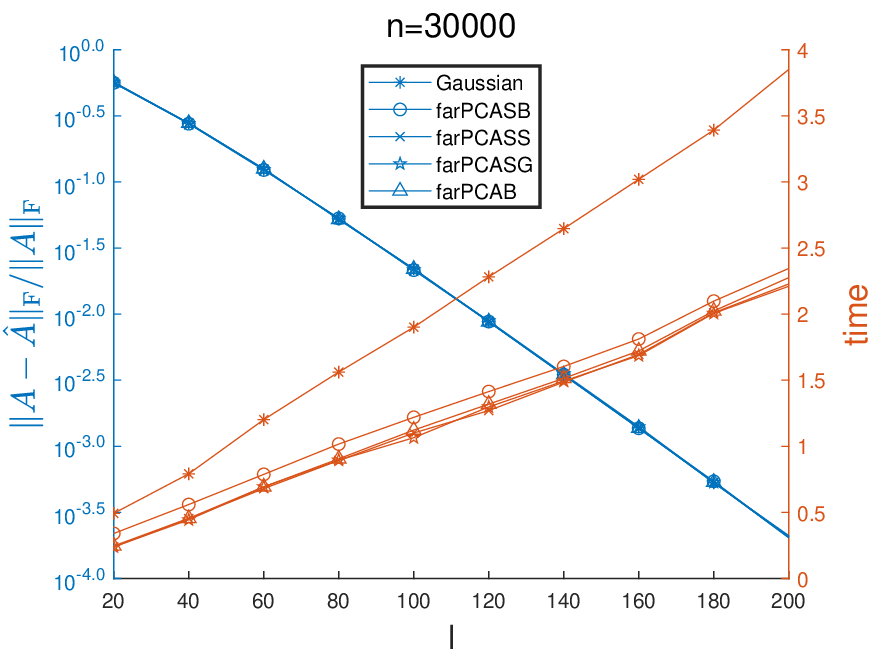}
	\end{minipage}
    \end{subfigure}

    \caption{Comparison of relative Frobenius-norm errors and execution times between Gaussian (\texttt{farPCA}), \texttt{farPCASB}, \texttt{farPCASS}, \texttt{farPCASG} and \texttt{farPCAB} with varying value of $l$ under $P=0$ (averaged over $20$ trials). The first two columns correspond to Matrix $1$ and Matrix $2$, and rows denote matrices of varying dimensions.}
   \label{fig2}
\end{figure}

As illustrated in Figure \ref{fig1}, the computational costs of \texttt{farPCASB}, \texttt{farPCASS}, \texttt{farPCASG}, and \texttt{farPCAB} consistently remain below those of \texttt{farPCA}, while achieving comparable error levels. The time advantages of the four algorithms are improved as $n$ increases, aligning with theoretical expectations in Remark \ref{remark3.4}. For instance, when $l$ is small and $P=1$, these four algorithms typically outperform \texttt{farPCA}, with \texttt{farPCASS}, \texttt{farPCASG}, and \texttt{farPCAB} achieving time savings of approximately $25\%$. Similarly, Figure \ref{fig2} indicates that with $P=0$ and $n=30000$, the time savings approach near $50\%$, further corroborating Remark \ref{remark3.4}. Simultaneously, the performances of \texttt{farPCASS}, \texttt{farPCASG}, and \texttt{farPCAB} are similar, with \texttt{farPCASS} and \texttt{farPCASG} preferred for their theoretical guarantees. Furthermore, as highlighted in subsections \ref{sec5.3} and \ref{sec5.5}, \texttt{farPCASG} exhibits slightly greater stability compared to \texttt{farPCASS} and \texttt{farPCAB}.

\subsection{Efficiency validation for solving the fixed-precision problems}\label{sec5.2}
In this subsection, we consider the termination condition with a relative tolerance $\|\mathbf{A}-\mathbf{Q}\mathbf{C}\|_\mathrm{F}\leq\varepsilon\|\mathbf{A}\|_\mathrm{F}$ for the fixed-precision problems. We still generate two types of input matrices $\mathbf{A}\in\mathbb{R}^{n\times n}$, where $n=5000$, $10000$, $30000$. The value of $p$ for each acceleration algorithm is set as in subsection \ref{sec5.1}. We choose block size $b=\min\{\max\{20,\lfloor n/100\rfloor\},50\}$, power parameter $P=1$, and the maximum number of iterations is set to $\lceil n/(2b)\rceil$. Here we use $\lceil x \rceil$ to denote the smallest integer greater than or equal to $x$, and $\lfloor x \rfloor$ to denote the largest integer less than or equal $x$. The average results of $20$ experiments under different input matrices, dimension size $n$, and thresholds $\varepsilon$ are delineated in Table \ref{table1}.

\begin{table*}[htbp]
\caption{Time, rank, and relative Frobenius-norm error results for fixed-precision simulations with different matrices, different dimensions $n$, and different desired thresholds $\varepsilon$. The measurements were averaged over $20$ trials. ``$l$", ``$t$", and ``$\varepsilon_{opt}$" represent the output rank, running time in seconds, and the relative Frobenius-norm error, respectively. For each setting, we highlight the fastest algorithm in bold.}
\resizebox{\linewidth}{!}{
\centering
\begin{tabular}{ccc|ccc|ccc|ccc|ccc|ccc}
\hline
\multirow{2}{*}{Matrix}&\multirow{2}{*}{$n$}&\multirow{2}{*}{$\varepsilon$}&\multicolumn{3}{c|}{Gaussian (\texttt{farPCA})} &\multicolumn{3}{c|}{\texttt{farPCASB}}&\multicolumn{3}{c|}{\texttt{farPCASS}} &\multicolumn{3}{c|}{\texttt{farPCASG}} &\multicolumn{3}{c}{\texttt{farPCAB}}\cr
 & & &$l$&$t$&$\varepsilon_{opt}$&$l$&$t$&$\varepsilon_{opt}$&$l$&$t$&$\varepsilon_{opt}$&$l$&$t$&$\varepsilon_{opt}$&$l$&$t$&$\varepsilon_{opt}$\cr
 \hline
 \multirow{6}{*}{Matrix 1}
 &\multirow{2}{*}{5000}&1e-4&350&0.25&9.02e-5&350&0.22&9.03e-5&350&\textbf{0.21}&9.02e-5&350&\textbf{0.21}&9.02e-5&350&0.22&9.02e-5\cr
 & &5e-5&550&0.43&4.58e-5&550&0.39&4.58e-5&550&\textbf{0.37}&4.58e-5&550&\textbf{0.37}&4.58e-5&550&0.38&4.58e-5\cr
 &\multirow{2}{*}{10000}&1e-4&350&0.89&9.02e-5&350&0.74&9.03e-5&350&\textbf{0.70}&9.10e-5&350&\textbf{0.70}&9.02e-5&350&0.72&9.02e-5\cr
 & &5e-5&550&1.44&4.58e-5&550&1.19&4.58e-5&550&\textbf{1.14}&4.58e-5&550&1.15&4.58e-5&550&1.18&4.58e-5\cr
 &\multirow{2}{*}{30000}&1e-4&350&7.49&9.03e-5&350&5.99&9.02e-5&350&5.87&9.03e-5&350&\textbf{5.86}&9.02e-5&350&5.92&9.02e-5\cr
 & &5e-5&550&11.76&4.58e-5&550&9.35&4.58e-5&550&\textbf{9.18}&4.58e-5&550&9.20&4.58e-5&550&9.28&4.58e-5\cr
\hline
\multirow{6}{*}{Matrix 2}
 &\multirow{2}{*}{5000}&1e-4&200&0.14&5.04e-5&200&0.12&5.03e-5&200&\textbf{0.11}&5.01e-5&200&\textbf{0.11}&4.99e-5&200&0.12&5.04e-5\cr
 & &5e-6&250&0.18&4.10e-6&250&\textbf{0.14}&4.11e-6&250&0.15&4.11e-6&250&\textbf{0.14}&4.14e-6&250&0.15&4.13e-6\cr
 &\multirow{2}{*}{10000}&1e-4&200&0.51&5.03e-5&200&0.42&5.04e-5&200&\textbf{0.40}&4.99e-5&200&\textbf{0.40}&5.03e-5&200&0.42&4.99e-5\cr
 & &5e-6&250&0.64&4.13e-6&250&0.52&4.14e-6&250&0.52&3.93e-6&250&\textbf{0.51}&4.14e-6&250&0.56&4.16e-6\cr
 &\multirow{2}{*}{30000}&1e-4&200&4.32&5.01e-5&200&3.50&5.05e-5&200&3.39&5.06e-5&200&\textbf{3.38}&5.02e-5&200&3.42&5.01e-5\cr
 & &5e-6&250&5.36&4.11e-6&250&4.31&4.14e-6&250&\textbf{4.19}&4.15e-6&250&\textbf{4.19}&4.13e-6&250&4.23&4.13e-6\cr
 \hline
\end{tabular}
}
\label{table1}
\end{table*}

According to Table \ref{table1}, the results align with those of the fixed-rank problem discussed in subsection \ref{sec5.1}, indicating that \texttt{farPCASB} is typically faster than \texttt{farPCA}. Meanwhile, \texttt{farPCASS}, \texttt{farPCASG}, and \texttt{farPCAB} demonstrate even greater acceleration and exhibit comparable performance. Furthermore, these algorithms consistently meet the required levels of accuracy for predefined values of $p$.

\subsection{Choice of the value of $p$}\label{sec5.3}

In this subsection, we still generate two types of input matrices as before and consider the termination condition with the relative tolerance $\|\mathbf{A}-\mathbf{Q}\mathbf{C}\|_\mathrm{F}\leq\varepsilon\|\mathbf{A}\|_\mathrm{F}$ for the fixed-precision problems. We examine how variations in parameter $p$ affect the performance of the four acceleration algorithms.

We present the results for Matrix $1$ only, as they are similar to those for Matrix $2$. We set $n=5000$, block size $b=50$, power parameter $P=1$, and $\varepsilon=5\times10^{-5}$. The maximum number of iterations is set to $\lceil\frac{n}{2b}\rceil$. The success rates and execution times averaged over $100$ times of \texttt{farPCASB}, \texttt{farPCASS}, \texttt{farPCASG}, and \texttt{farPCAB} for Matrix $1$ across varying values of $p$ are presented in Figure \ref{fig3}, where $p=10^{-3}$, $ln(n)/n$, $2\times10^{-3}$, $5\times10^{-3}$, $10^{-2}$, and $10^{-1}$.

Figure \ref{fig3} indicates that extremely small $p$ values could result in larger iterations as discussed in subsection \ref{sec5.1}, leading to increased computational cost. \texttt{farPCASB} presents the highest stability for sufficiently small $p$, while \texttt{farPCASS} and \texttt{farPCASG}, and \texttt{farPCAB} exhibit comparable performance but are relatively unstable under sufficiently small $p$. However, as $p$ increases, these algorithms stabilize. For $p=2\times10^{-3}$, $5\times10^{-3}$, $10^{-2}$, and $10^{-1}$, the success rates approach $100\%$, indicating these algorithms are not sensitive to larger $p$. Therefore, $p_1=\max\{10^{-3},\ln(n)/n\}$ is recommended for \texttt{farPCASB}, and $p_2=\max\{10^{-3},10/n\}$ is recommended for the others empirically. Under the pre-defined values of $p$, \texttt{farPCASG} exhibits the fastest speed.
\begin{figure}[htbp]
    \begin{subfigure}{\textwidth}
	\centering
	\begin{minipage}{0.31\textwidth}
			\centering
			\includegraphics[width=\textwidth]{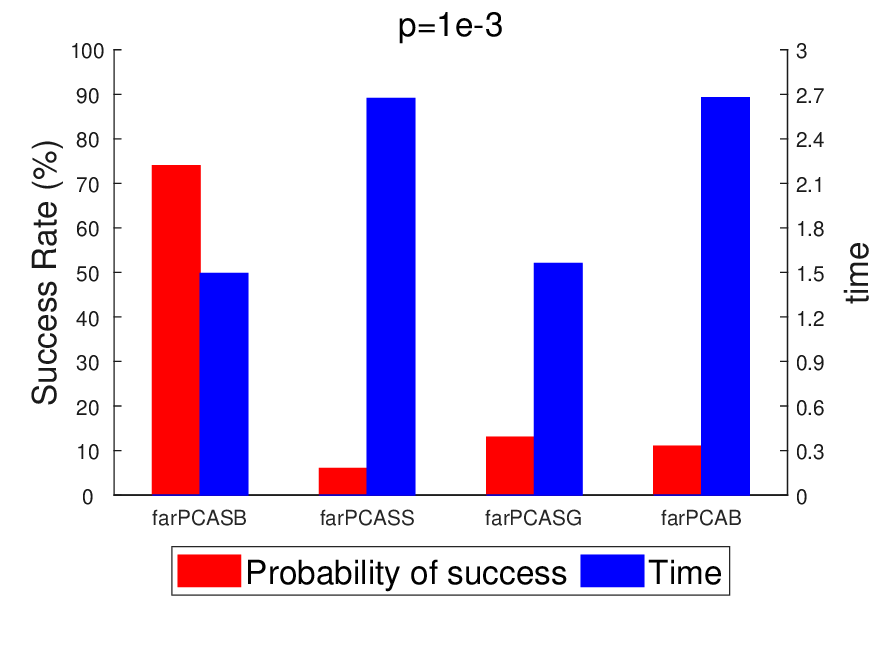}
		\end{minipage}
		\begin{minipage}{0.31\textwidth}
			\centering
			\includegraphics[width=\textwidth]{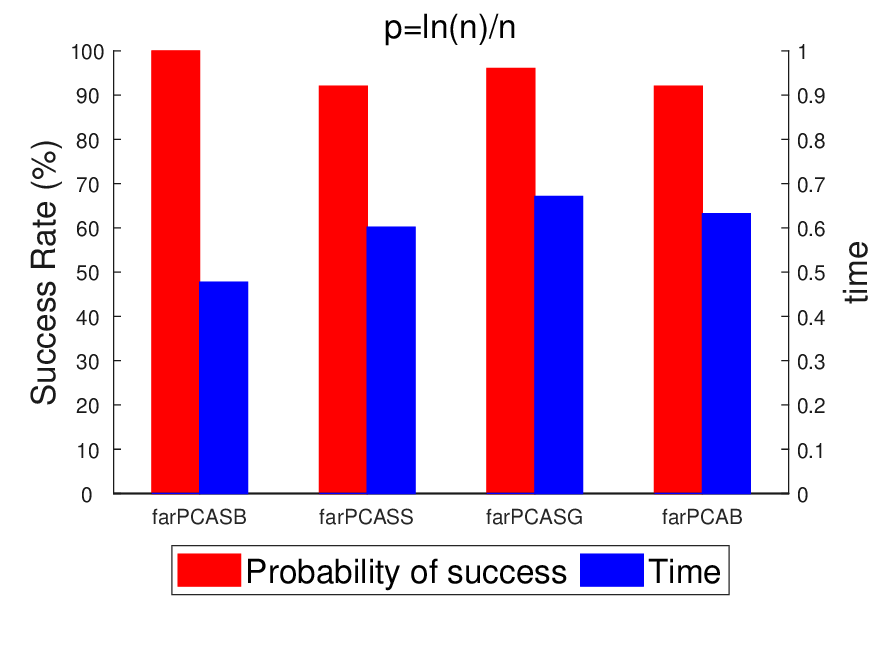}
	\end{minipage}
 \begin{minipage}{0.31\textwidth}
			\centering
			\includegraphics[width=\textwidth]{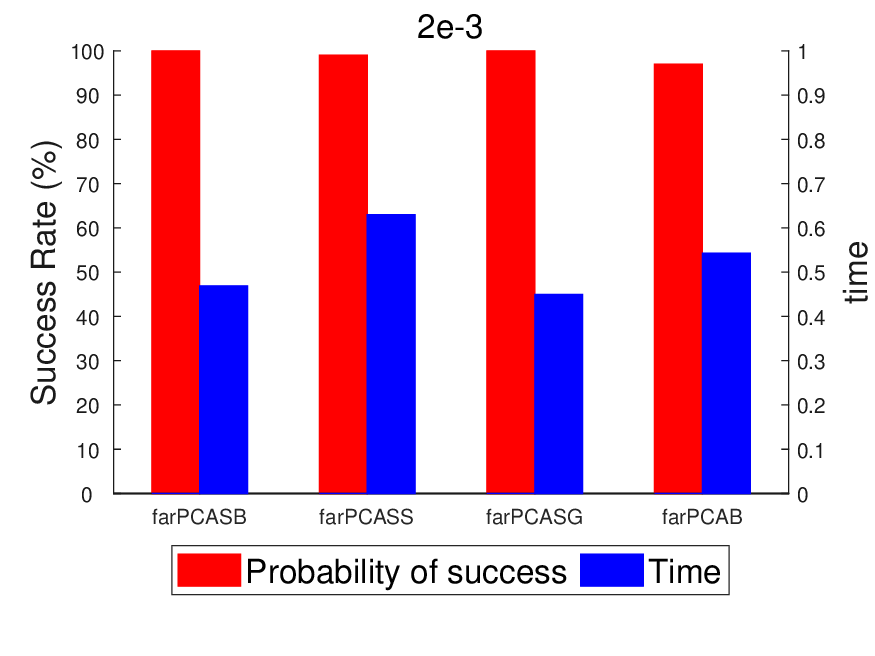}
	\end{minipage}

    \begin{minipage}{0.31\textwidth}
			\centering
			\includegraphics[width=\textwidth]{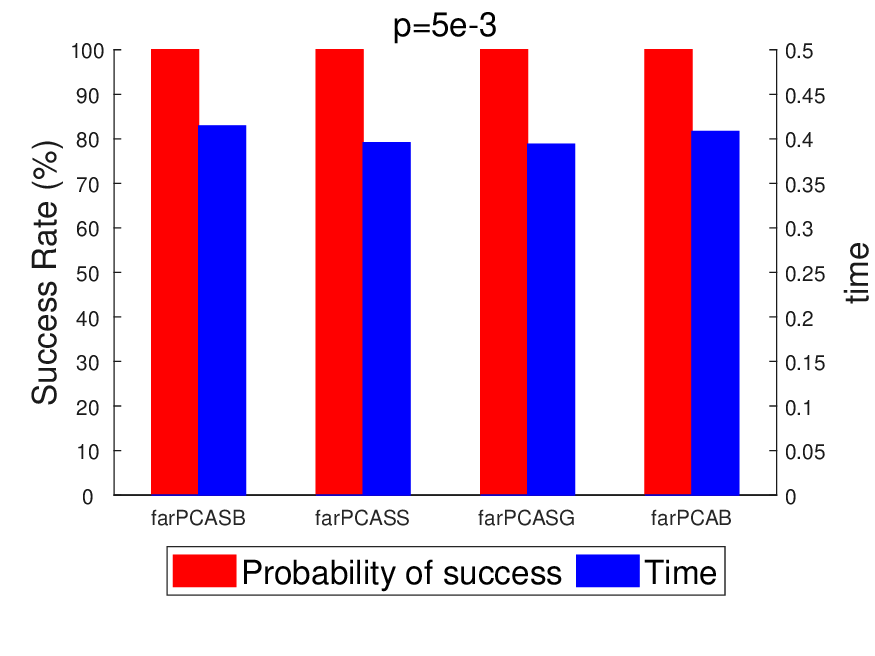}
	\end{minipage}
    \begin{minipage}{0.31\textwidth}
			\centering
			\includegraphics[width=\textwidth]{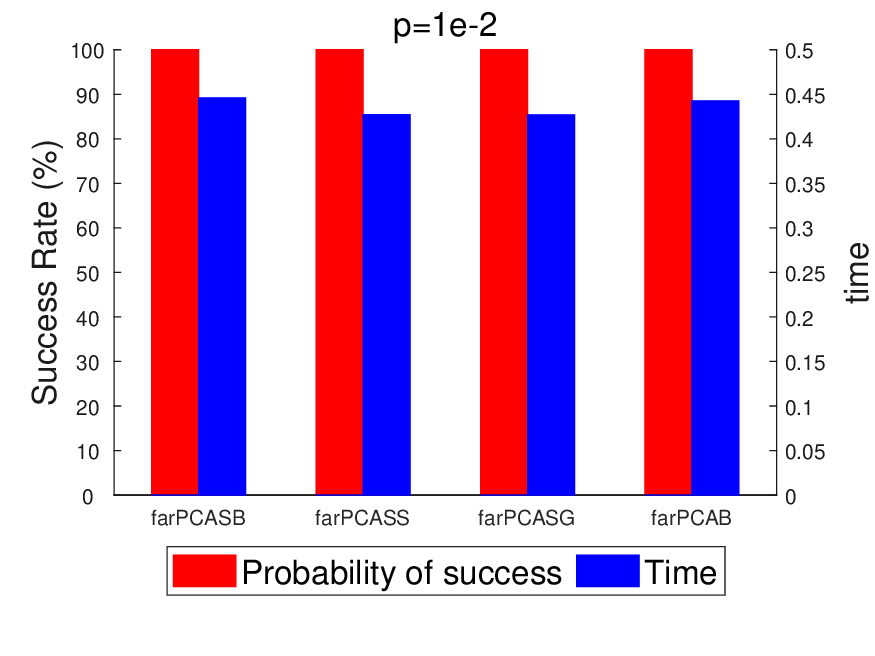}
	\end{minipage}
    \begin{minipage}{0.31\textwidth}
			\centering
			\includegraphics[width=\textwidth]{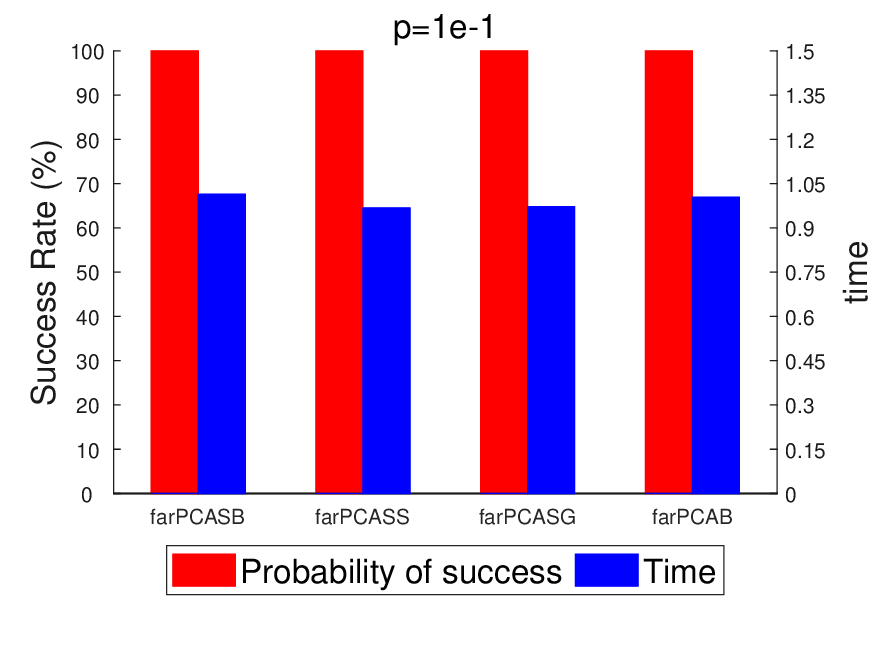}
	\end{minipage}
    \end{subfigure}
   	\caption{Success rates and execution times of \texttt{farPCASB}, \texttt{farPCASS}, \texttt{farPCASG}, and \texttt{farPCAB} for Matrix $1$ sized $5000\times 5000$ across varying values of $p$ in $100$ trials.}
   \label{fig3} 
\end{figure}

\subsection{Results for real data}\label{sec5.4}
We assess \texttt{farPCA}, \texttt{farPCASB} with $p_1=\max\{10^{-3},\ln(\max\{m,n\})/\max\{m,n\}\}$, \texttt{farPCASS}, \texttt{farPCASG}, and \texttt{farPCAB} with $p_2=\max\{10^{-3},10/\max\{m,n\}\}$, \texttt{randUBV}, and \texttt{svds} on real data, where $m,n$ are the dimensions of the target matrix. The real data is a $3168\times 4752\times 3$ dense matrix, depicting the RGB image of a spruce pine. By reorganizing this matrix's structure into a $9504\times 4752$ matrix, we compare these algorithms with the threshold $\varepsilon=0.1$. As described in \cite{hallman2022block}, for the fixed-precision problem, after applying these algorithms to attain the approximate SVD, the post-processing strategy can be adopted to truncate the SVD of $\mathbf{A}$ to the lowest rank that fulfills the specified accuracy threshold. We set block size $b=\min\{\max\{20,\lfloor\min(m, n)/100\rfloor\},50\}$ and power parameter $P=1$ or $5$. The maximum number of iterations is set to $\lceil{n}/{2b}\rceil$. The results from a single execution are shown in Table \ref{table2}, where $l$ and $r$ denote the rank produced by algorithms and the truncated rank after the post-processing, respectively. Given that \texttt{svds} is incapable of addressing the fixed-precision factorization problem, we use $r$ computed by \texttt{farPCA} with $P = 5$ as the input rank ($l$) for \texttt{svds}, record its execution time, and further determine the optimal rank $r$.

\begin{table}[htbp]
\caption{Results for image data with $\varepsilon=0.1$. The terms ``$t_{dec}$" ``$t_{post}$" ``$t_{total}$" mean the running time required for algorithms to decompose, the time for post-processing and obtaining the final approximate SVD steps, and the total time, with all time reported in seconds. ``$\varepsilon_{opt}$" represents the relative Frobenius-norm error, and $l$ and $r$ denote the rank produced by algorithms and the truncated rank after the post-processing, respectively. Outside of \texttt{svds}, the fastest algorithm and the fastest near-optimal algorithm (the smallest $r$) are highlighted in bold, with speedup factors of the accelerated algorithms relative to \texttt{farPCA} shown in parentheses.}
\scriptsize
\begin{center}
\begin{tabular}{cccccccc}
\hline
\multicolumn{2}{c}{Method} &$t_{dec}$ &$t_{post}$ &$t_{total}$ &$l$ &$r$ &$\varepsilon_{opt}$ \cr
\hline
\multicolumn{2}{c}{\texttt{randUBV}}&0.51 &0.04 &0.55 &611 &451 &9.99e-2\cr
\hline
\multirow{5}{*}{$P=1$}&Gaussian (\texttt{farPCA}) &0.61 &0.01 &0.62 &470 &467 &9.99e-2\cr
 &\texttt{farPCASB} &0.51 &0.01 &0.52(16.1\%) &470 &467 &9.997e-2\cr
 &\texttt{farPCASS} &0.50 &0.01 &\textbf{0.51}(17.7\%) &470 &467 &9.997e-2\cr
 &\texttt{farPCASG} &0.53 &0.01 &0.54(12.9\%) &470 &467 &9.998e-2\cr
 &\texttt{farPCAB} &0.51 &0.01 &0.52(16.1\%) &470 &467 &9.997e-2\cr
\hline
\multirow{5}{*}{$P=5$}&Gaussian (\texttt{farPCA}) &1.70 &0.01 &1.71 &470 &427 &9.99e-2\cr
 &\texttt{farPCASB} &1.57 &0.01 &\textbf{1.58}(7.6\%) &470 &427 &9.99e-2\cr
 &\texttt{farPCASS} &1.57 &0.01 &\textbf{1.58}(7.6\%) &470 &427 &9.99e-2\cr
 &\texttt{farPCASG} &1.57 &0.01 &\textbf{1.58}(7.6\%) &470 &427 &9.99e-2\cr
 &\texttt{farPCAB} &1.61 &0.01 &1.62(5.3\%) &470 &427 &9.99e-2\cr
\hline
\multicolumn{2}{c}{\texttt{svds}} &-- &-- &28.00 &427 &426 &9.98e-2\cr
\hline
\end{tabular}    
\end{center}
\label{table2}
\end{table}

The results again validate that the proposed algorithms not only meet the accuracy requirements automatically but also typically demonstrate significant efficiency compared to \texttt{farPCA}. In comparison to \texttt{randUBV}, these algorithms incur similar or even lower computational costs while maintaining the flexibility to achieve nearly optimal approximation results. Moreover, these proposed algorithms can produce near-optimal outcomes with significantly lower computational costs than \texttt{svds}, while possessing the capacity to automatically terminate the algorithm once the approximation error reaches an anticipated accuracy.

\subsection{Stability}\label{sec5.5}
We now explore cases where condition \eqref{eq3.1} is unmet and examine algorithm stability with respect to the probability parameter $p$. We still construct $\mathbf{A}$ as the form of Matrix $1$, except that $\mathbf{V}_A$ is a block diagonal matrix with block $\mathbf{V}_i,i=1,\cdots,d$ being $(n/d)\times(n/d)$ orthogonal matrices attained by orthonormalizing the standard Gaussian matrix. Here $d$ is a parameter that measures the sparsity of $\mathbf{V}_A$-it is diagonal when $d=n$ and dense when $d=1$. Results are delineated in Table \ref{table4}.
\begin{table}[htbp]
\resizebox{\linewidth}{!}{
\centering
\begin{tabular}{c|ccc|ccc|ccc|ccc}
\hline
\multirow{2}{*}{$p$} &\multicolumn{3}{c|}{\texttt{farPCASB}}&\multicolumn{3}{c|}{\texttt{farPCASS}} &\multicolumn{3}{c|}{\texttt{farPCASG}} &\multicolumn{3}{c}{\texttt{farPCAB}}\cr &$d_1$&$d_2$&$d_3$&$d_1$&$d_2$&$d_3$&$d_1$&$d_2$&$d_3$&$d_1$&$d_2$&$d_3$\cr
 \hline
$ln(5000)/5000$&100&63 &37 &99 &58 &39 &94 &65 &42 &95 &63 &38\cr
2e-3&100&79 &37 &98 &68 &49 &100&74 &56 &98 &71 &43\cr
1e-2&100&100&100&100&100&100&100&100&95 &100&100&92\cr
\hline
\end{tabular}
}
\caption{Success rates across different $d$ values ($d_1=1$, $d_2=500$, $d_3=5000$) and $p$ values for a $5000\times5000$ matrix in $100$ simulations for the fixed-precision problem with threshold $\varepsilon=1\times 10^{-4}$. The block size, power parameter, and the maximum number of iterations are the same as in subsection \ref{sec5.4}.}
\label{table4}
\end{table}

Table \ref{table4} further shows that these acceleration algorithms remain effective for larger values of 
$p$, with \texttt{farPCASB} demonstrating the highest stability, followed by \texttt{farPCASG}, which outperforms \texttt{farPCASS} and \texttt{farPCAB} in terms of stability. Therefore, we recommend \texttt{farPCASB} for maximum robustness and \texttt{farPCASG} for a balance between efficiency and stability.

\section{Concluding Remark}\label{sec6}
This paper introduces the standardized Bernoulli, sparse sign, and sparse Gaussian matrices as replacements for the standard Gaussian matrix in \texttt{farPCA}, aiming to improve computational efficiency in low-rank approximation. These alternative matrices require less computational cost during matrix-matrix multiplications and, under mild conditions, converge in distribution to a standard Gaussian matrix when multiplied by an orthogonal matrix. Therefore, both theoretically and empirically, the three corresponding proposed algorithms maintain performance equivalent to \texttt{farPCA} asymptotically while achieving faster computation, making them a superior alternative to \texttt{farPCA}. Among these algorithms, we recommend \texttt{farPCASB} for maximum robustness and \texttt{farPCASG} for a balance between efficiency and stability. Fortunately, the cases that hinder algorithm efficiency, which occur when condition \eqref{eq3.1} is not satisfied, are rarely encountered in practical applications. A related criterion from \cite{candes2012exact} suggests that for successful matrix completion, the singular vectors of the target matrix should be well-distributed across all dimensions.

Future research will focus on determining whether the right singular matrix of the target matrix is sufficiently "well-defined" for use with the acceleration algorithms. Additionally, random matrices with i.i.d. elements (mean 0, variance 1) may exhibit similar properties, prompting further exploration of matrices with enhanced performance. We also plan to study adaptive strategies for tuning the parameter $p$ across iterations to balance computational efficiency and stability. Finally, integrating these methods with advanced techniques and applying them to machine learning and related fields will be a key area of focus.

\section*{Funding} \noindent This work was funded by the Key technologies for coordination and interoperation of power distribution service
  	resource, Grant No. 2021YFB2401300.

\bibliographystyle{spmpsci}
\bibliography{sn-bibliography}
\end{document}